\documentclass{article}
\pdfoutput=1
\usepackage[margin=1.25in]{geometry}
\usepackage{amsmath,amssymb,amsthm,bbm,multirow,asymptote,graphicx,theoremref}
\usepackage{mathtools,accents}
\usepackage{authblk}
\numberwithin{equation}{section}
\theoremstyle{plain}
\newtheorem{theorem}{Theorem}[section]
\newtheorem*{theorem*}{Theorem}
\newtheorem{lemma}[theorem]{Lemma}

\theoremstyle{remark}
\newtheorem{remark}[theorem]{Remark}

\theoremstyle{definition}
\newtheorem{definition}{Definition}[section]
\newtheorem*{definition*}{Definition}
\DeclareMathOperator{\sgn}{sgn}

\newcommand{\E}{\mathbb{E}}
\renewcommand{\P}{\mathbb{P}}
\newcommand{\Q}{\mathbb{Q}}
\newcommand{\di}{\mathrm{d}}
\newcommand{\iden}{\boldsymbol{1}}
\newcommand{\half}{\dfrac{1}{2}}
\newcommand{\R}{\mathbb{R}}
\newcommand{\D}{\mathcal{D}}
\newcommand{\taud}{\tau_{\mathcal{D}}}

\newcommand{\dx}{\mathrm{d}x}

\newcommand{\F}{\mathcal{F}}
\newcommand{\Ep}[1]{\E\left[#1\right]}
\newcommand{\Eps}[2]{\E^{#1}\left[#2\right]}

\newcommand{\pjt}{p_{j,t}}
\newcommand{\qjt}{q_{j,t}}

\newcommand{\ttau}{\tilde{\tau}}
\newcommand{\e}{\mathrm{e}}
\newcommand{\eps}{\varepsilon}

\begin{document}

\title{Robust Hedging of Options on a Leveraged Exchange Traded Fund}

\author{Alexander M. G. Cox}
\author{Sam M. Kinsley}
\affil{Department of Mathematical Sciences, University of Bath, U.K.}

\maketitle

\begin{abstract}
A leveraged exchange traded fund (LETF) is an exchange traded fund that uses financial derivatives to amplify the price changes of a basket of goods. In this paper, we consider the robust hedging of European options on a LETF, finding model-free bounds on the price of these options.

To obtain an upper bound, we establish a new optimal solution to the Skorokhod embedding problem (SEP) using methods introduced in Beiglb\"ock-Cox-Huesmann. This stopping time can be represented as the hitting time of some region by a Brownian motion, but unlike other solutions of e.g. Root, this region is not unique. Much of this paper is dedicated to characterising the choice of the embedding region that gives the required optimality property. Notably, this appears to be the first solution to the SEP where the solution is not uniquely characterised by its geometric structure, and an additional condition is needed on the stopping region to guarantee that it is the optimiser. An important part of determining the optimal region is identifying the correct form of the dual solution, which has a financial interpretation as a model-independent superhedging strategy.
\end{abstract}


\section{Introduction}
Given a Brownian motion $B$ and a centered probability distribution $\mu$ on the real line which has finite second moment, the Skorokhod embedding problem is to find a stopping time $\tau$ such that
\begin{equation}\
B_{\tau} \medspace \text{has law} \medspace \mu\medspace \text{and} \medspace (B_{t\wedge\tau})_{t\geq0} \text{ is UI}. \tag{SEP} \label{SEP}
\end{equation}
In this paper we give a solution to this problem which has the property that it maximises $\Ep{F(B_{\tau},\tau)}$ over solutions of \eqref{SEP} for a certain function $F$ that has the financial motivation of being the payoff of a European call option on a leveraged exchange traded fund. In Section \ref{existence} we show the existence of such a stopping time using the monotonicity principle of \cite{Beiglboeck:2013aa}. This solution can be seen as the hitting time of a region we call a $K$-cave barrier, which is the combination of a Root barrier and a Rost inverse barrier separated by a curve $K(x)$. It is well known that for such a distribution $\mu$, there is a Root barrier such that the hitting time of that barrier solves \eqref{SEP}, and moreover that barrier is unique. However, it is easy to see that in most cases there will be infinitely many $K$-cave barrier solutions of \eqref{SEP}.

The main difficulty which arises in this problem is then to determine which of these solutions is optimal, and much of this paper is dedicated to finding a necessary and sufficient condition that ensures we have the optimal stopping region. In Section \ref{heuristics} we propose such a condition using a heuristic PDE argument, and then we confirm that this condition is sufficient in Section \ref{optimality} using probabilistic arguments. To do this we introduce the dual problem of finding the minimal starting cost of a model-independent superhedging strategy, and use martingale theory to derive an expression for the optimal dual solution.

To argue the converse, that is, there is at least one dual optimiser satisfying this condition, we take a different approach. In \cite{Cox:2016ab} we set up a linear programming problem which is a discretised optimal Skorokhod embedding problem and for which we can prove a strong duality result. The motivation of \cite{Cox:2016ab} was to help determine the form of our dual superhedging strategy in this continuous time problem, and indeed the strong duality result gives the existence of dual optimisers. In \cite{Cox:2016ab} we show that we can recover our continuous time problem as the limit of the discrete linear programming problems, and in Section \ref{discrete} of this paper we show that our superhedging strategy is the limit of the discrete dual optimisers. We then verify that our proposed condition is both necessary and sufficient.

As well as the financial relevance of this problem, we also believe that the solution we give to the Skorokhod embedding problem is theoretically important. In \cite{Beiglboeck:2013aa}, it was shown that \emph{every} known solution to \eqref{SEP} which possessed an optimality criteria could be derived as a consequence of the monotonicity principle. Specifically, the monotonicity principle implies a geometric structure that is sufficient to \emph{uniquely} identify the stopping region. The construction we provide in this paper uses the monotonicity principle to deduce important geometric structure of the solution, but this does not uniquely characterise the resulting stopping region, and we therefore need to provide an additional criterion which specifies which of the possible stopping regions we should choose. To the best of our knowledge, this is the first example of such a condition in the literature.


\subsection{Background}
The standard approach to pricing and hedging exotic options is to suppose the existence of some probabilistic model, and then determine the discounted, risk-neutral expected payoff under this model.  An alternative method is to use commonly-traded options, with prices that we can observe, to construct a hedging strategy. Usually we assume that we can observe call prices for a fixed maturity and multiple strikes, and identify models consistent with these prices. We can then attempt to find \textquoteleft model-independent', or \textquoteleft robust', bounds on our option price by finding the extremal feasible models. A least upper bound on the price of the option should be the smallest amount of money with which it is possible to maintain a super-replicating portfolio under any model. In other words, we are required to give a portfolio which is a superhedge for all feasible models, and for which there is a specific model that gives the correct option prices, under which our superhedge is actually a hedge, i.e. we replicate the option exactly. This approach, although it doesn't give a single arbitrage-free price for the option, has the obvious advantage that it eliminates model risk.

A result of Breeden and Litzenberger in \cite{Breeden:1978aa} says that given European call prices of all positive strikes for some fixed maturity $T$, we can calculate the marginal distribution of the underlying, $S$, at this time $T$. Moreover, this distribution is given under the measure used by the market to price these options, so, unlike in the traditional methods, we do not need to change measure. In particular, if the call option prices $C(K)$ are calculated as the discounted expected payoff under a probability measure $\Q$, then, under certain arbitrage conditions on $C$, we have that
\begin{equation*}
C'_+(K)=-\Q(S_T>K).
\end{equation*}
A consistent model $(\Omega, \F, \P)$ must then be such that $S_T\sim \mu$ under $\P$, where $\mu$ is determined from the above Breeden-Litzenberger Formula. If we assume that the price process is a true martingale, by a time change this condition becomes equivalent to the Skorokhod embedding problem, \eqref{SEP}. It is also important to note that given a solution $\tau$ to (SEP),
\begin{equation*}
M_t=B_{t/{(T-t)\wedge\tau}}
\end{equation*}
is a martingale with $M_T\sim\mu$, and in fact there is a one-to-one correspondence between solutions of \eqref{SEP} and (uniformly integrable) martingales $M$ with $M_T\sim\mu$.

There are many solutions to \eqref{SEP}, see \cite{Obloj:2004aa} for a survey of all solutions known at the time, some of which have nice optimality properties. For example, the solution of Root \cite{Root:1969aa} was shown by Rost in \cite{Rost:1976aa} to minimise $\Ep{F(\tau)}$ for convex $F$ over solutions of \eqref{SEP}. The Root stopping time is the hitting time of a region known as a barrier, and in this paper we find a solution of \eqref{SEP} which can also be viewed as the hitting time of a Brownian Motion of a certain region. This stopping time has the property that it maximises the expected payoff of a certain exotic option. Using the Skorokhod embedding problem to find no-arbitrage prices of exotic options given prices of vanilla options was first developed by Hobson in \cite{Hobson:1998aa}, and since used and extended in many works, including \cite{Brown:2001aa, Cox:2015ab, Cox:2008aa, Cox:2011aa, Cox:2013ac, Cox:2013ab, Henry-Labordere:2016aa, Hobson:2012aa, Hobson:2002aa}. In particular we refer to the survey article of Hobson \cite{Hobson:2011aa}.

As mentioned above, the \textquoteleft primal' problem of finding a model-independent upper bound on an option price has a related \textquoteleft dual' problem of finding the smallest amount of money with which a self-financing model-independent superhedging portfolio can be funded. There are duality results on various options under suitable conditions, for example \cite{Cox:2014aa, Dolinsky:2014aa, Dolinsky:2014ab, Dolinsky:2015aa, Hou:2015aa}. We follow the pathwise inequality approach of Cox and Wang \cite{Cox:2013ac, Cox:2013ab} to show a duality result and find the form of the minimal superhedging portfolio.

Our problem\footnote{We are grateful to Pierre Henri-Labord\`ere for suggesting this problem to us.} is motivated by the pricing of a call option on a leveraged exchange traded fund, LETF, in particular we look at finding an arbitrage-free upper bound on the price of these options. An exchange traded fund (ETF) is a security traded on a stock exchange that tracks an index or basket of assets. An ETF is an ownership stake in a pool of assets, so a number of investors can share in a large, diverse portfolio, spreading the transaction costs across all investors. A regular ETF matches the benchmark index's performance 1:1, whereas a leveraged ETF will most commonly match it 2:1 or 3:1, usually by holding daily futures contracts. Daily compounding means that LETFs do not maintain their 2:1 relative performance over time, only over a single day, and even then transaction costs and fees need to be subtracted. For example, if we have a traditional index ETF and a 2:1 LETF both trading at \$100, and the index increases by 10\% that day, then our ETF is at \$110 whilst our LETF is now worth \$120. Our LETF met its goal on this individual day, but then these prices are now fixed, since our funds are compounded daily. If the following day our index sees a decrease of 9\%, then the ETF is at \$100.10, but our LETF value decreases by 18\% to \$98.4. We can clearly see that over time we will not maintain our 2:1 ratio. 

The first LETF was released in 2006, and in 2016 there are over 200 LETFs available, most commonly with 125\%, 200\%, or 300\% ratios. At the time of writing, the value of assets in the global ETF market is over \$3 trillion, and some investors expect it to double in size by 2021. LETFs are typically written on very liquid ETFs, with vanilla options traded on both the ETF and the LETF. This means that our assumption of observing European call option prices on the underlying ETF is a reasonable one. LETFs have been studied mathematically in recent literature, for example \cite{Ahn:2015aa, Avellaneda:2010aa, Cheng:2009aa, Zhang:2010aa}. In particular, in \cite{Zhang:2010aa}, Zhang considers options on an LETF in terms of options on the underlying ETF, giving a closed form solution when the volatility of $\log (S_t)$ is deterministic, and numerical results fitting various models when the volatility is random.

\subsection{Formulation}
We will work in continuous time, thus assuming the LETF portfolio is rebalanced continuously. Let $S_t$, $L_t$ be the prices of the ETF and LETF respectively, and suppose $S$ is some continuous martingale. Setting interest rates and transaction costs to 0 and renormalising, the dynamics of the LETF with leverage ratio $\beta>1$ are given by
\begin{equation*}
L_t=S_t^\beta\exp\left(-\frac{\beta(\beta-1)}{2}V_t\right),
\end{equation*}
where $V_t$ is the accumulated quadratic variation of $\log S_t$ up to time $t$. It is easy to verify that $L_t$ is a martingale when $S_t$ is. To avoid dealing with the accumulated log quadratic variation, we time change by setting $\tau_t:=\inf\{s\geq0:\medspace V_s=t\}$ and $X_t:=S_{\tau_t}$. But then,
\begin{equation*}
\mathrm{d}\langle X\rangle_t=\mathrm{d}\langle S\rangle_{\tau_t}=S_{\tau_t}^2\mathrm{d}V_{\tau_t}=X_t^2\mathrm{d}t
\end{equation*}
and therefore $X_t$ is a geometric Brownian motion (GBM). The payoff function for a European call option on the time-changed LETF with strike $k>0$ is
\begin{equation}
\label{case2}
F_{L}(x,t)=\left(x^\beta\exp\left(-\frac{\beta(\beta-1)}{2}t\right)-k\right)_+.
\end{equation}
Write $h_{L}(x,t)=x^\beta\exp\left(-\frac{\beta(\beta-1)}{2}t\right)$ so that $h_{L}(X_t,t)$ is a martingale since $X_t$ is.

The problem of finding an upper bound on the price of such an option is then equivalent to solving the optimal Skorokhod embedding problem
\begin{equation*} \label{LOptSEP}\tag{LOptSEP}
\sup_{\tau} \Ep{\left(X_{\tau}^\beta\exp\left(-\frac{\beta(\beta-1)}{2}\tau\right)-k\right)_+} \quad \text{over stopping times $\tau$ such that } X_{\tau}\sim\mu,
\end{equation*}
where $X$ is a GBM, and in fact an exponential martingale. We conjecture a hitting time solution with a stopping region of the form shown in Figure~\ref{spike}, bounded by curves $l_{L}(x)$ and $r_{L}(x)$ giving the boundary of an inverse-barrier and a barrier region respectively (defined below), such that $l_{L}(x)\leq K_{L}(x)\leq r_{L}(x)$ and $l_{L}$ is increasing. The curve $K_{L}(x)=\frac{2}{\beta(\beta-1)}\ln(\frac{x^\beta}{k})$ is such that $h_{L}(x,K_{L}(x))=k$, so we only \textquoteleft score' a positive payoff if we are absorbed by $l_{L}$, i.e. to the left of $K_{L}$. The example in Figure~\ref{spike} contains an infinite section, and the barriers could also have spikes, we assume no differentiability on the curves $l_{L},$ $r_{L}$.

We show that the function $r_{L}$ is such that $\underline{\mathcal{R}}=\{(x,t):\thinspace t\geq r_{L}(x)\}$ is a barrier, i.e. a closed subset of $(-\infty,+\infty)\times[0,+\infty)$ such that if $(x,t)\in \underline{\mathcal{R}}$, then $(x,s)\in \underline{\mathcal{R}}$ for all $s>t$. Since we are working with geometric Brownian motion, we will have $(0,t)\in\underline{\mathcal{R}}$ for all $t$. Note in particular that the closedness of $\underline{\mathcal{R}}$ implies that $r_{L}$ is lower semi-continuous.

Similarly, $l_{L}$ is such that $\overline{\mathcal{R}}=\{(x,t):\thinspace t\leq l_{L}(x)\}$ is an inverse barrier, or reverse barrier, meaning a closed subset of $(-\infty,+\infty)\times[0,+\infty)$ such that if $(x,t)\in R$, then $(x,s)\in R$ for all $s<t$. Similarly to above, the function $l_L$ is upper semi-continuous. It is well known that the Rost embedding (\cite{Chacon:1985aa, Rost:1971aa}) is a solution of \eqref{SEP} which takes the form of the hitting time of an inverse barrier, see also \cite{Cox:2015aa, Cox:2013ac, Obloj:2004aa}.

We will call a stopping region of this form a $K$-cave barrier.

\begin{figure}[h]
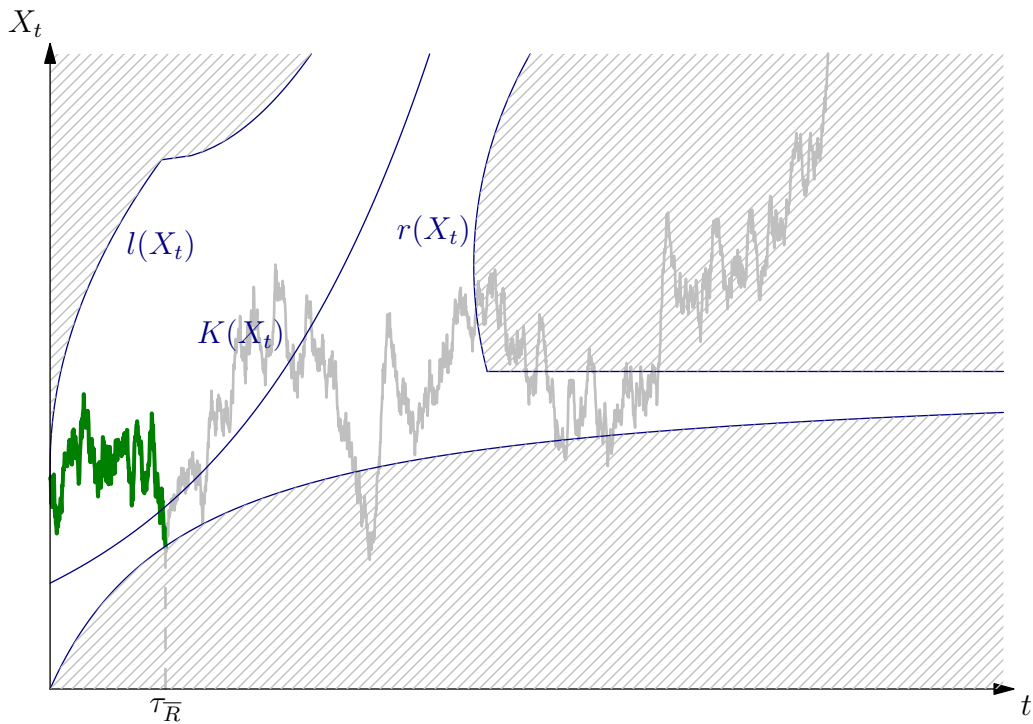

\centering
\begin{asy}[width=0.9\textwidth]
    import graph;
    import stats;
    import patterns;
    
     // We construct a Brownian motion of time length T, with N time-steps
    int N = 3000;
    //int N = 300;
    real T = 4.5;
    real dt = T/N;
    real B0 = 1;
    
    real xmax = 3;
    real xmin = 0;
    real r1 = 1.5;
    real l1 = 2.5;
    real x1 = 0.4;
    real x2 = 0.2;
    real eps = 0.05;
    real eps1 = 0.05;
    real kint = 0.5;

    real[] B; // Brownian motion
    real[] t; // Time

    path BM;
    
    // Seed the random number generator. Delete for a "random" path:
    srand(11);

    B[0] = B0;
    t[0] = 0;

    BM = (t[0],B[0]);
 // Define a barrier 

    real R1(real y) {return 2*exp(((y-2)**2)/8);}
    real R2(real y) {return (1/(r1-y)-1/r1);}
    real R2inv(real z) {return r1 - 1/(z+1/r1) ;}
    real R(real y) {return (y>r1) ? R1(y) : R2(y);}
    real L1(real y) {return  (y>B0+eps) ? ((y<l1) ? ((y-B0-eps)**2)/4 : sqrt(y-l1)+((1.45^2)/4)) : 0 ;}
    real L(real y) {return (y>B0) ? L1(y) : 0 ;}
    real K(real y) {return (y>0) ? log(y)-log(kint) : 0 ;}

    int H = N+1;
    int H2 = N+1;
    int BMstop;
    int BMstop2 = N+2;
    
      for (int i=1; i<N+1; ++i)
    {
      B[i] = B[i-1] + Gaussrand()*sqrt(dt);
      t[i] = i*dt;
      BM = BM--(t[i],B[i]);
    
          if ((H==N+1)&&(t[i]>=R(B[i])))
      {
	H = i;
	BMstop = length(BM);
      }

      if ((H==N+1)&&(t[i]<L(B[i])))
      {
	H = i;
	BMstop = length(BM);
      }
      if ((H2==N+1) && (B[i]<0))
      {
	B[i] = 0;
	BMstop2 = length(BM);
	H2 = i;
      }	
      if ((H2==N+1) && (B[i]>=xmax))
      {
	B[i] = xmax;
	BMstop2 = length(BM);
	H2 = i;
      }			
    }
    
    if (H==N+1)
    BMstop = length(BM);

    pen p = deepgreen + 1.5;
    pen p2 = mediumgray + 1;

    if (H<N+1)
    draw(subpath(BM,BMstop,BMstop2),p2);

    draw(subpath(BM,1,BMstop),p);

    pair tau = point(BM,BMstop);

    draw((tau.x,0)--tau,p2+dashed);
    label("$\tau_{\overline{R}}$",(tau.x,0),S);

    pen q = black + 0.5;
    
    draw("$l(X_t)$",graph(L1,identity,B0,xmax),SE,deepblue+0.5);
    // draw((L1(r1),r1)--(T,r1),NW,deepblue+0.5);
    draw(graph(R2,identity,xmin,R2inv(T)),NW,deepblue+0.5);
    draw("$r(X_t)$",graph(R1,identity,r1,xmax),SW,deepblue+0.5);
    draw((R1(r1),r1)--(T,r1),NW,deepblue+0.5);
    draw("$K(X_t)$", graph(K,identity,kint,xmax),NW,deepblue+0.5);

    path barrier1 = (graph(L1,identity,xmax,B0+eps)--(0,xmax)--(L(xmax),xmax)--cycle);

    add("hatch",hatch(1mm,mediumgray));

    fill(barrier1,pattern("hatch"));

    draw((0,xmin)--(0,xmax+eps1),q,Arrow);
    draw((0,xmin)--((T+eps1),xmin),q,Arrow);
    label("$t$",(T+eps1,xmin),SE);
    label("$X_t$",(0,(xmax+eps1)),NW);
    
     path barrier3 = (graph(R1,identity,xmax,r1)--(R1(r1),r1)--(T,r1)--(T,xmax)--cycle);

    path barrier4 = (graph(R2,identity,xmin,R2inv(T))--(T,xmin)--cycle);

    add("hatch",hatch(1mm,mediumgray));

    fill(barrier3,pattern("hatch"));

    fill(barrier4,pattern("hatch"));

    \end{asy}
\caption{An example of our LETF problem that has a $K$-cave barrier with an infinite region}
\label{spike}
\end{figure}

In the Brownian case, there is a unique Root barrier, or Rost inverse barrier, that solves \eqref{SEP} for any centred distribution $\mu$ with finite second moment (and no atom at $0$ for the Rost case), and these solutions have the nice property that they minimise, or maximise, the expected value of any convex function of the stopping time. In \cite{Beiglboeck:2013aa} the authors introduce a new embedding, the cave embedding, which can be viewed as the combination of a Root and a Rost embedding. Our stopping region has a similar form, and much of the analysis in this paper also applies to the cave embedding. In particular, the results we derive in Section~\ref{dualconvergence} can be deduced for the cave embedding in essentially the same manner as in this paper. The following results can also be adapted to price European call options on an inverse LETF, i.e. where $\beta<0$, and the only difference is the shape of the curve $K$.

We will actually consider pricing two options, the first of which is the problem described by \eqref{LOptSEP}.  The second problem is very similar and is notable due to its structure as a European call option on an exponential martingale. For this case, we consider the payoff function
\begin{equation}
\label{F}
F_{BM}(x,t)=\left(\exp\left(\beta x-\frac{1}{2}\beta^2t\right)-k\right)_+
\end{equation}
for $\beta>0$ a constant, and $k>0$ our strike. Here we can think of $B$, the discounted price process, as a Brownian motion (BM) after a time change. We define $h_{BM}$ to be $h_{BM}(x,t)=\exp(\beta x-\frac{1}{2}\beta^2t)$ so that $h_{BM}(B_t,t)$ is a martingale, and we have a similar stopping region given by $l_{BM}(x), r_{BM}(x)$ seperated by $K_{BM}(x)=\frac{2x}{\beta}-\frac{2}{\beta^2}\ln(k)$, as shown in Figure~\ref{nospike}. Our problem in this case is 
\begin{equation*} \label{OptSEP}\tag{OptSEP}
\sup_{\tau} \Ep{\left(\exp\left(\beta B_{\tau}-\frac{1}{2}\beta^2\tau\right)-k\right)_+} \quad \text{over solutions of \eqref{SEP}}.
\end{equation*}

It will usually be clear which problem we are talking about, in which case we will drop the subscripts. We will alternate which case we use in the proofs of results to give a clear representation of both, but all of our results hold for both problems. In fact, the problems are closely related since
\begin{equation*}
X_t=X_0\exp\left(B_t-\half t\right) \implies X_t^{\beta}\exp\left(-\frac{\beta(\beta-1)}{2}\tau\right)=X_0\exp\left(\beta B_t-\half \beta^2 t\right).
\end{equation*}
However, the embedding condition applies to different processes, and this is where the problems differ.

\begin{figure}[h]
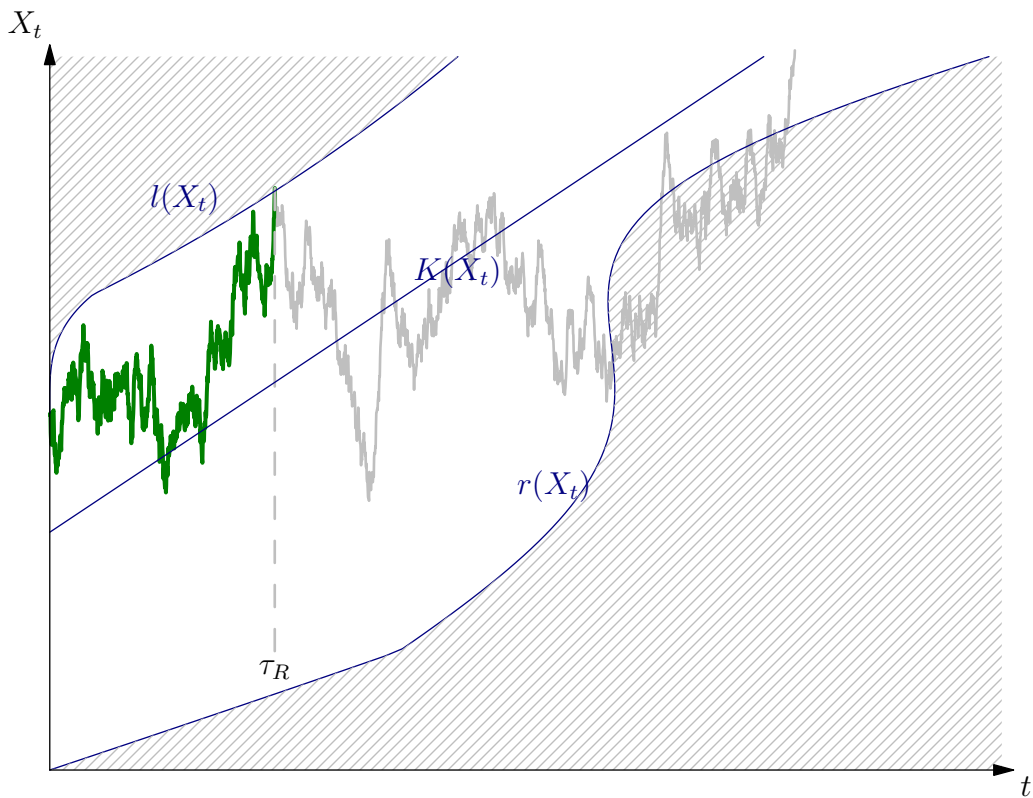

\begin{asy}[width=0.9\textwidth]
    import graph;
    import stats;
    import patterns;
 
    // We construct a Brownian motion of time length T, with N time-steps
    int N = 3000;
    //int N = 300;
    real T = 4;
    real dt = T/N;
    real B0 = 1;
    
    real xmax = 2.5;
    real xmin = -0.5;
    real l1 = 1.5;
    real x1 = 0.4;
    real x2 = 0.2;
    real eps = 0.05;
    real eps1 = 0.05;
    real kint = 0.5;

    real[] B; // Brownian motion
    real[] t; // Time

    path BM;
    
    // Seed the random number generator. Delete for a "random" path:
    srand(11);

    B[0] = B0;
    t[0] = 0;

    BM = (t[0],B[0]);
 // Define a barrier 
 
   real R(real y) {return (y>0) ? (2*(-exp(-(y-2)**2)+0.8*y+exp(-6.25)+0.75)) : (2*(-exp(-(y-2)**2)+1.5*y+exp(-6.25)+0.75)) ;}
    real L1(real y) {return  (y>B0+eps) ? ((y<l1) ? 2((y-B0-eps)^3) : 3(log(y)-log(l1)+(2/3)*(l1-B0-eps)^3)) : 0 ;}
    real L(real y) {return (y>B0) ? L1(y) : 0 ;}
    real K(real y) {return (2*(0.75*y-0.75*kint)) ;}

int H = N+1;
    int H2 = N+1;
    int BMstop;
    int BMstop2 = N+2;
    
      for (int i=1; i<N+1; ++i)
    {
      B[i] = B[i-1] + Gaussrand()*sqrt(dt);
      t[i] = i*dt;
      BM = BM--(t[i],B[i]);
      
      if ((H==N+1)&&(t[i]>=R(B[i])))
      {
        H=i;
        BMstop = length(BM);
      }   
      if ((H==N+1)&&(t[i]<L(B[i])))
      {
	H = i;
	BMstop = length(BM);
      }
      if ((H2==N+1) && (B[i]<-1))
      {
	B[i] = 0;
	BMstop2 = length(BM);
	H2 = i;
      }	
      if ((H2==N+1) && (B[i]>=xmax))
      {
	B[i] = xmax;
	BMstop2 = length(BM);
	H2 = i;
      }			
    }
    
    if (H==N+1)
   { BMstop = length(BM);}

    pen p = deepgreen + 1.5;
    pen p2 = mediumgray + 1;

    if (H<N+1)
    {draw(subpath(BM,BMstop,BMstop2),p2);}

    draw(subpath(BM,1,BMstop),p);

    pair tau = point(BM,BMstop);

    draw((tau.x,0)--tau,p2+dashed);
    label("$\tau_{R}$",(tau.x,0),S);

    pen q = black + 0.5;
    
    draw("$l(X_t)$",graph(L1,identity,B0,xmax),NW,deepblue+0.5);
    draw("$r(X_t)$", graph(R,identity,xmin,xmax),SW,deepblue+0.5);
    draw("$K(X_t)$", graph(K,identity,kint,xmax),NE,deepblue+0.5);

    path barrier1 = (graph(L1,identity,xmax,B0+eps)--(0,xmax)--(L(xmax),xmax)--cycle);

    add("hatch",hatch(1mm,mediumgray));

    fill(barrier1,pattern("hatch"));
    
    path barrier2 = (graph(R,identity,xmax,xmin)--(R(xmin),xmin)--(T,xmin)--(T,xmax)--cycle);
    
    add("hatch",hatch(1mm,mediumgray));

    fill(barrier2,pattern("hatch"));

    draw((0,xmin)--(0,xmax+eps1),q,Arrow);
    draw((0,xmin)--((T+eps1),xmin),q,Arrow);
    label("$t$",(T+eps1,xmin),SE);
    label("$X_t$",(0,(xmax+eps1)),NW);

    \end{asy}
\caption{An example $K$-cave stopping region with continuous boundaries}
\label{nospike}
\end{figure}

\section{Existence of a Maximiser} \label{existence}
\subsection{Stop-Go Pairs}
In \cite{Beiglboeck:2013aa}, the authors introduce the idea of \emph{stop-go pairs}, and develop a monotonicity principle that allows them to prove, using a unified approach, the existence of all solutions to (SEP) that have an optimality property. The intuition behind stop-go pairs is as follows: we consider continuous paths starting from 0 and want to decide when it is optimal to stop or continue them. Consider a stopped path $(g,t)$ and a path that is not yet stopped $(f,s)$, where $f(s)=g(t)$. We imagine stopping $(f,s)$ at time $s$ and creating a continuation of $(g,t)$ by transferring all paths which extend $(f,s)$ onto $(g,t)$. If this improves the value of the quantity we are optimising, then we have contradicted the optimality of the stopping region. In this case we call $((f,s),(g,t))$ a stop-go pair, and we denote the set of stop-go pairs by $\textsf{SG}$. This then can be extended by a second optimality problem in order to sort the pairs that see exactly the same value of the optimality problem when mass is transferred onto the stopped path.

Formally, \cite{Beiglboeck:2013aa} considers $\mathsf{S}=\{(f,s):f(s)\rightarrow \mathbb{R} \medspace \text{is continuous},f(0)=0\}$ and a Borel function $\gamma :\mathsf{S}\rightarrow \mathbb{R}$, so $\gamma_t=\gamma((X_s)_{s\leq t},t)$ is an optional stochastic process. Our problem is to find the maximiser of 
\begin{equation}
\label{Opt1}
P_{\gamma}=\sup\{\mathbb{E}[{\gamma}_\tau]: \tau \medspace \text{solves (SEP)}\}.
\end{equation}
We set
\begin{equation*}
(\gamma^{\oplus (f,s)})_u:=\gamma (f\oplus B, s+u),
\end{equation*}
and then $(f,g)$ is a stop-go pair, $(f,g)\in \textsf{SG}$, if for every stopping time $\sigma$ such that $0<\mathbb{E}[\sigma]<\infty$,
\begin{equation*}
\mathbb{E}[(\gamma^{\oplus (f,s)})_\sigma] + \gamma(g,t) < \gamma(f,s) + \mathbb{E}[(\gamma^{\oplus (g,t)})_\sigma].
\end{equation*}
If $\hat{\tau}$ is our maximiser, we can then find a set $\Gamma\subseteq \mathsf{S}$ with $\mathbb{P}[((B_s)_{s\leq\hat{\tau}},\hat{\tau})\in \Gamma]=1$, such that $\Gamma$ is $\gamma$-monotone, that is,
\begin{equation*}
\textsf{SG}\cap\left(\Gamma^< \times \Gamma\right)=\varnothing,
\end{equation*}
where $\Gamma^<:=\{(f,s):\exists (\tilde{f},\tilde{s})\in\Gamma,s<\tilde{s} \medspace \text{and} f\equiv\tilde{f}\medspace \text{on}\medspace [0,s]\}$. Denote the set of maximisers of $P_{\gamma}$ by $\mathsf{Opt}_{\gamma}$ and consider another Borel function $\tilde{\gamma}:\mathsf{S}\rightarrow\mathbb{R}$. In \cite{Beiglboeck:2013aa} it is shown that $\mathsf{Opt}_\gamma$ is non-empty and compact for suitable $\gamma$, and so we can assume that $\hat{\tau}$ is also a maximiser of the secondary optimisation problem
\begin{equation}
\label{Opt2}
P_{\tilde{\gamma}|\gamma}=\sup\{\mathbb{E}[\tilde{\gamma_\tau}]:\tau\in \mathsf{Opt}_\gamma\}.
\end{equation}
The set of secondary stop-go pairs, $\textsf{SG}_2$ consists of all $((f,s),(g,t)))\in \mathsf{S}\times \mathsf{S}$ such that $f(s)=g(t)$ and for every stopping time $\sigma$ with $0<\mathbb{E}[\sigma]<\infty$ we have
\begin{equation}
\label{SGleq}
\mathbb{E}[(\gamma^{\oplus (f,s)})_\sigma] + \gamma(g,t) \leq \gamma(f,s) + \mathbb{E}[(\gamma^{\oplus (g,t)})_\sigma]
\end{equation}
and the equality
\begin{equation}
\label{SGeq}
\mathbb{E}[(\gamma^{\oplus (f,s)})_\sigma] + \gamma(g,t) = \gamma(f,s) + \mathbb{E}[(\gamma^{\oplus (g,t)})_\sigma]
\end{equation}
implies the inequality
\begin{equation}
\label{SG2}
\mathbb{E}[(\tilde{\gamma}^{\oplus (f,s)})_\sigma] + \tilde{\gamma}(g,t) < \tilde{\gamma}(f,s) + \mathbb{E}[(\tilde{\gamma}^{\oplus (g,t)})_\sigma].
\end{equation}
Then we can also assume that
\begin{equation*}
\textsf{SG}_2 \cap \left(\Gamma^< \times \Gamma\right)=\varnothing.
\end{equation*}
Theorem 7.1 in \cite{Beiglboeck:2013aa} tells us that there exists a $\gamma$-monotone Borel set $\Gamma\subseteq \mathsf{S}$ such that $\mathbb{P}$-a.s. $((B_t)_{t\leq\hat{\tau}},\hat{\tau})\in\Gamma$.

\subsection{Existence Theorem}
We use the notion of stop-go pairs to prove the following theorem:
\begin{theorem} \thlabel{SGTheorem}
There exists a stopping time $\tau_{\mathcal{R}}$ which maximises $\mathbb{E}[F(B_\tau,\tau)]$ over all solutions to \eqref{SEP} and which is of the form $\tau_{\mathcal{R}}=\inf\{t>0:(X_t,t)\in \mathcal{R}\}$ for some $K$-cave barrier $\mathcal{R}$.
\end{theorem}

To prove this we consider the set of stop-go pairs of our primary and secondary, yet to be determined, optimality problems, and for these we need to introduce local times. The local time of a continuous semimartingale $X$ at $a$ is the increasing, continuous process $L^a$ that gives the It\^o-Tanaka formula:
\begin{equation*}
(X_t -a)_+=(X_0 -a)_+ + \int_0^t \mathbbm{1}_{\{X_s> a\}}\mathrm{d}X_s + \frac{1}{2}L_t^a.
\end{equation*}
Observe that we can write
\begin{equation*}
L_t^{a}(X)=\lim_{\epsilon\to 0}\frac{1}{\epsilon}\int_0^t \mathbbm{1}_{[a,a+\epsilon]}(X_s)\mathrm{d}\langle X\rangle_s,
\end{equation*}
where $\langle X\rangle_s$ is the quadratic variation process of $X$. By establishing the form of $\textsf{SG}$ and $\textsf{SG}_2$ we can argue as in \cite{Beiglboeck:2013aa} that we have a $\gamma$-monotone set supporting a maximiser of our two optimisation problems, and that it can be written as a stopping time of the required form.

\begin{proof}
We write
\begin{align*}
M_u^f&:=h(f(s)+B_u,s+u)=\exp(\beta(f(s)+B_{u})-\frac{\beta^2}{2}(s+u))=h^fM_u\\
M_u^g&:=h(g(t)+B_u,t+u)=\exp(\beta(g(t)+B_{u})-\frac{\beta^2}{2}(t+u))=h^gM_u
\end{align*}
where the constant $h^f=h(f(s),s)$ is introduced to emphasise that the process is of the form $(h(f(s),s)M_u)_u$, for $M_u=\exp(\beta B_u-\frac{\beta^2}{2}u)$ with $(B_u)_u$ a BM. We then see that, for $\alpha=\frac{h^g}{h^f}=\exp\left(-\frac{1}{2}\beta^2(t-s)\right)$,
\begin{equation*}
M_u^g=\alpha M_u^f.
\end{equation*}
Then, after applying the Monotone Convergence Theorem along the localising sequence $\sigma_j=\sigma\wedge j$, along with Fatou's Lemma and Conditional Jensen's Inequality, the first term in \eqref{SGleq} becomes
\begin{equation*}
\mathbb{E}[F(f\oplus B,s+\sigma)]=(h^f-k)_+ +\frac{1}{2}\mathbb{E}[L_{\sigma}^k(M^f)].
\end{equation*}
Here we are taking the local time of the process $(M^f_u)_u$ accumulated at $k$ up to time $\sigma$.


If we use the It\^o-Tanaka formula on both sides of \eqref{SGleq}, this is equivalent to
\begin{equation*}
(h^f-k)_+ +\frac{1}{2}\mathbb{E}[L_{\sigma}^k(M^f)]+(h^g-k)_+\leq (h^g-k)_+ +\frac{1}{2}\mathbb{E}[L_{\sigma}^k(M^g)]+(h^f-k)_+,
\end{equation*}
which holds iff
\begin{equation*}
\mathbb{E}[L_{\sigma}^k(M^f)]\leq \mathbb{E}[L_{\sigma}^k(M^g)].
\end{equation*}
We have equality, i.e. case \eqref{SGeq}, when
\begin{equation*}
\mathbb{E}[L_{\sigma}^k(M^f)]= \mathbb{E}[L_{\sigma}^k(M^g)],
\end{equation*}
which clearly holds when $h^f=h^g$, which happens exactly when $s=t$, since $f(s)=g(t)$.

We aim to show that
\begin{equation}
\label{casesless}
h^f<h^g<k\implies\begin{cases}
\text{either} & \mathbb{E}[L_{\sigma}^k(M^f)]< \mathbb{E}[L_{\sigma}^k(M^g)]\\
\text{or} & \mathbb{E}[L_{\sigma}^k(M^f)]= \mathbb{E}[L_{\sigma}^k(M^g)]=0\end{cases}
\end{equation}
\begin{equation}
\label{casesmore}
h^f>h^g>k\implies\begin{cases}
\text{either} & \mathbb{E}[L_{\sigma}^k(M^f)]<\mathbb{E}[L_{\sigma}^k(M^g)]\\
\text{or} & \mathbb{E}[L_{\sigma}^k(M^f)]= \mathbb{E}[L_{\sigma}^k(M^g)]=0.\end{cases}
\end{equation}

We have to argue the two cases seperately, so suppose first that $h^f<h^g<k$ and take a stopping time $\sigma$ such that $\mathbb{P}[M^g_{\sigma}>k]>0$. Also let $\alpha=\frac{h^g}{h^f}=\exp(-\frac{\beta^2}{2}(t-s))$, so $\alpha>1$. Then we have
\begin{align*}
\mathbb{E}[L_{\sigma}^k(M^g)]&=\mathbb{E}[|M^g_{\sigma}-k|]-|h^g-k|\\
&=\mathbb{E}[|M^g_{\sigma}-k|-k]+h^g\\
&=\mathbb{E}[-M^g_{\sigma}+2(M^g_{\sigma}-k)_+]+h^g\\
&>\mathbb{E}\left[-M^g_{\sigma}+2\left(M^g_{\sigma}-\alpha k\right)_+\right]+h^g\\
&=\mathbb{E}\left[\left|M^g_{\sigma}-\alpha k\right|\right]-\left|h^g-\alpha k\right|\\
&=\mathbb{E}[L_{\sigma}^{\alpha k}(M^g)].
\end{align*}

Now we note (see for example \cite[Chapter~VI, Exercise~1.22]{Revuz:1999aa}) that if $f$ is a strictly increasing function that can be written as the difference of two convex functions, $a>0$, and $X_t$ a continuous semimartingale,
\begin{equation*}
L_t^{f(a)}(f(X))=f_+'(a)L_t^a(X).
\end{equation*}
 We apply the result with $f(k)=\alpha k$ to find $L_{\sigma}^{\alpha k}(M^g)=\alpha L_{\sigma}^k(M^f)$.

Then, combining the last two results, taking expectations, and noting that $\alpha>1$, we see that if $\mathbb{E}[L_{\sigma}^k(M^g)]>0$, so that $\mathbb{P}[M^g_{\sigma}>k]>0$, then
\begin{equation*}
\mathbb{E}[L_{\sigma}^k(M^g)]>\mathbb{E}\left[L_{\sigma}^{\alpha k}(M^g)\right]=\alpha\mathbb{E}[L_{\sigma}^k(M^f)]>\mathbb{E}[L_{\sigma}^k(M^f)].
\end{equation*}

This gives \eqref{casesless}, but we require a different argument for \eqref{casesmore}, so suppose now that $k<h^g<h^f$. We have $M^g_u=h^g \mathrm{e}^{Y_u}$, for $Y_u=\beta B_u - \half \beta^2 u$, so using the local time result above, $L^k_{\sigma}(M^g)=k L^{\log{\frac{k}{h^g}}}_{\sigma}(Y)$, and similarly $L^k_{\sigma}(M^f)=k L^{\log{\frac{k}{h^f}}}_{\sigma}(Y)$. This means that our problem is equivalent to considering the local time spent by Brownian motion with drift at two different levels. Consider the function $U(x,y)=\Eps{x}{L^y_\infty (Y)}$. By the strong Markov property, 
\begin{equation*}
U(x,y)=U(y,y)\P^x\left(H_y(Y)<\infty\right),
\end{equation*}
where $H_z(Y)=\inf\{u\geq 0 : \medspace Y_u=z\}$, so
then we know that
\begin{list}{$\bullet$}{}
\item $U(x,y)$ is constant for $y\leq x$: $U(x,y)=U(x,x)$ $\forall y\leq x$
\item $U(x,y)$ is strictly decreasing in $y$ for $y\geq x$.
\end{list}
In \eqref{SGleq} we run our Brownian motion until a stopping time $\sigma$, so let $\nu=\mathcal{L}\left(Y_\sigma\right)$ and suppose $\sigma$ is such that $\mathbb{P}^x[\sigma>H_{a'}(Y)]>0$, where $x>a>a'$. Then, by the properties of $U$ above,
\begin{align*}
\Eps{x}{L_\sigma^a (Y)}&=\Eps{x}{L_\infty^a(Y)} - \Eps{Y_\sigma}{L_\infty^a(Y)} \\
&= U(x,a) - \int_\R U(y,a) \nu(\di y) \\
&= U(x,a') - \int_\R U(y,a) \nu(\di y) \\
&> U(x,a') - \int_\R U(y,a') \nu(\di y) \\
&= \Eps{x}{L_\sigma^{a'} (Y)}.
\end{align*}
Setting $a=\log{\frac{k}{h^g}}$ and $a'=\log{\frac{k}{h^f}}$ gives \eqref{casesmore}.

We now have that
\begin{equation*}
\textsf{SG}\supseteq\{((f,s),(g,t)):h^f>h^g>k\text{ or } h^f<h^g<k \text{ and } \mathbb{E}[L_{\sigma}^k(M^g)]>0 \text{ for all }\sigma\}
\end{equation*}
and the pairs in $\{((f,s),(g,t)):h^f>h^g>k\text{ or }h^f<h^g<k\}$ that are not in \textsf{SG} are those for which we can find a stopping time such that the expected values of the local times at $k$ up to the stopping time of the two processes are equal. However we have shown that if this is the case (and $s\neq t$) then these expected values must be equal to zero. This tells us that when we set our paths off at $h^f$ and $h^g$, they never reach $k$, and so in particular $\sgn(M^f_\sigma-k)=\sgn(h^f-k)$ and $\sgn(M^g_\sigma-k)=\sgn(h^g-k)$, and this also holds for all times up to $\sigma$. We now define our secondary optimality problem as in \eqref{Opt2} with
\begin{equation*}
\tilde{\gamma}(f,s)=-((f(s)-k)^+)^2+((f(s)-k)^-)^2
\end{equation*}
Consider a pair of paths $((f,s),(g,t))$ and a stopping time $\sigma$ such that $h^f<h^g<k$ and $L^k_{\sigma}(M^f)=L^k_{\sigma}(M^g)=0$. Substituting these into \eqref{SG2} gives
\begin{equation*}
\mathbb{E}[(k-M^f_\sigma)^2]+(k-h^g)^2<\mathbb{E}[(k-M^g_\sigma)^2]+(k-h^f)^2
\end{equation*}
which, by It\^o-Tanaka, simplifies to
\begin{equation*}
 \mathbb{E}[\langle M^f\rangle_\sigma]<\mathbb{E}[\langle M^g\rangle_\sigma].
\end{equation*}
This is true since $h^f<h^g$, and we finally have that 
\begin{align*}
\textsf{SG}_2&\supseteq\{((f,s),(g,t)):f(s)=g(t),\medspace s<t<\frac{2}{\beta^2}(\beta f(s)-\log(k)) \text{ or }s>t>\frac{2}{\beta^2}(\beta f(s)-\log(k))\}\\
&=\{((f,s),(g,t)):f(s)=g(t),\medspace h(f(s),s)<h(g(t),t)<k \text{ or } h(f(s),s)>h(g(t),t)>k\}.
\end{align*} 

Now, by \cite{Beiglboeck:2013aa}, that there exists a $\gamma$-monotone set $\Gamma\in \mathsf{S}$ with $\mathbb{P}[((B_s)_{s\leq \tau_{R}},\tau_{R})\in\Gamma]=1$, and we can complete our proof.

We know that there is a maximiser, $\tau_R$ of $P_\gamma$ and $P_{\tilde{\gamma}|\gamma}$, and that we can pick a $\gamma$-monotone set $\Gamma\in \mathsf{S}$ supporting $\tau_R$. Define
\begin{gather*}
\underline{\mathcal{R}}_{CL}=\{(m,x):\exists (g,t)\in\Gamma,\medspace h(g(t),t)\leq m<k, g(t)=x\}\\
\underline{\mathcal{R}}_{OP}=\{(m,x):\exists (g,t)\in\Gamma,\medspace h(g(t),t)<m<k, g(t)=x\}\\
\overline{\mathcal{R}}_{CL}=\{(m,x):\exists (g,t)\in\Gamma,\medspace h(g(t),t)\geq m>k, g(t)=x\}\\
\overline{\mathcal{R}}_{OP}=\{(m,x):\exists (g,t)\in\Gamma,\medspace h(g(t),t)>m>k, g(t)=x\}
\end{gather*}
and write $\mathcal{R}_{OP}=\underline{\mathcal{R}}_{OP}\cup \overline{\mathcal{R}}_{OP}$ and $\mathcal{R}_{CL}=\underline{\mathcal{R}}_{CL}\cup \overline{\mathcal{R}}_{CL}$. Denote the corresponding hitting times (by $(M_t(\omega),B_t(\omega))$) of these sets by $\tau_{OP}=\overline{\tau}_{OP}\wedge\underline{\tau}_{OP}$, and similarly for $\tau_{CL}$. We claim that $\tau_{CL}\leq \tau_R \leq \tau_{OP}$, and indeed we immediately see that by the definition of $\mathcal{R}_{CL}$ we have that $\tau_{CL}\leq \tau_R$.

To show the second inequality pick $\omega$ such that $((B_s)_{s\leq\tau_{R}(\omega)},\tau_R(\omega))\in\Gamma$ and assume for contradiction that $\overline{\tau}_{OP}(\omega)<\tau_R(\omega)$ (the argument for $\underline{\tau}_{OP}(\omega)$ is similar). Then $\exists s\in[\overline{\tau}_{OP}(\omega),\tau_R(\omega))$ such that $f:=(B_r(\omega))_{r\leq s}$ has $(h(f(s),s),f(s))\in\overline{\mathcal{R}}_{OP}$. Since $s<\tau_R(\omega)$ we know that $f\in\Gamma^<$. But then by the definition of $\overline{\mathcal{R}}_{OP}$, $\exists (g,t)\in\Gamma$ such that $f(s)=g(t)$ and $h(g(t),t)>h(f(s),s)>k$ which contradicts the $\gamma$-monotonicity of $\Gamma$, since $(g(t),f(s))\in \textsf{SG}_2 \cap (\Gamma^< \times \Gamma)$.

Finally, observe that for $\underline{\tau}$ by the Strong Markov Property, and the fact that one-dimensional Brownian Motion immediately returns to its starting point, that $\underline{\tau}_{CL}=\underline{\tau}_{OP}$. For $\overline{\tau}$ we argue as in the Rost embedding case of \cite[Theorem~2.4]{Beiglboeck:2013aa}.

It is clear that we then have such a domain consisting of a barrier and an inverse barrier seperated by $K(x)$, since when $f(s)=g(t)$ we have that $h^f>h^g\geq m$ $\implies$ $s<t$.
\end{proof}

\begin{remark}
To repeat these arguments for the LETF payoff we have $h(x,t)=x^\beta\exp(-\frac{\beta(\beta-1)}{2}t)$, and instead of $M^f$ and $M^g$ we look at
\begin{align*}
X_u^f&:=h(f(s)+X_u,s+u)=(f(s)+X_u)^\beta\exp\left(-\frac{\beta(\beta-1)}{2}(s+u)\right)\\
X_u^g&:=h(g(t)+X_u,t+u)=(g(t)+X_u)^\beta\exp\left(-\frac{\beta(\beta-1)}{2}(t+u)\right)
\end{align*}
where $(X_u)_u$ is a GBM. For the inverse-barrier argument we have that, since $f(s)=g(t)$, $X_u^g=\alpha X_u^f$ for $\alpha=\exp(-\frac{\beta(\beta-1)}{2}(t-s))$. For $k<h^g<h^f$ we write $X_u^g=h^g\e^{Y_u}$ where $Y_u$ is again a martingale with a negative drift. We can then repeat exactly the arguments above.
\end{remark}


\subsection{Non-uniqueness} \label{nonunique}
We have proven that there is a solution to \eqref{OptSEP} which maximises our expected terminal payoff and is the hitting time of a $K$-cave barrier, but it is important to note that there is not a unique solution to \eqref{SEP} of this form for non-trivial distributions. 

One example of non-uniqueness is a result of having a non-increasing left-hand boundary $l$. In this case there can be areas of $l$ that we do not hit, and so these parts of $l$ could actually take any form, as long as they do not embed any mass. Any choice of $l$ has an increasing equivalent (where on any regions we do not hit, $l$ just remains constant), and to remove this form of non-uniqueness we can assume that we are taking this choice of the left boundary. This is equivalent to the idea of uniqueness of regular barriers, as introduced by Loynes in \cite{Loynes:1970aa}.

Even once we have made this choice of $l$, a more troublesome form of non-uniqueness can occur. Consider for example an atomic distribution with atoms at three points $N$, $-N$, and $z\in (0,N)$. We can change our stopping time by moving the points $r(z)$ and $l(z)$, and it is easy to see that increasing $l(z)$ (moving our left hand boundary at $z$ to the right) increases the amount of mass stopped at $z$. Similarly, decreasing $l(z)$ decreases the amount of paths stopped by this boundary, and moving $r(z)$ has the opposite effect. For certain distributions $\mu$ we will be able to move $l(z)$ and $r(z)$ slightly (in the same direction, without crossing $K$) and still embed $\mu(\{z\})$. Note that the amounts embedded at $N$ and $-N$ will not change since we know that the total mass embedded must sum to one, and we also have a martingale condition on our process, and these two conditions fix the masses embedded at these extreme points once the amount stopped at $z$ is fixed. We could therefore have multiple barriers embedding $\mu$ each with a different stopping time, and therefore a different payoff, so we need some condition on the barriers that gives us the optimal choice. To find this condition we will consider the dual problem.

It is well known that the dual problem of maximising an expected payoff is to find the minimum cost with which we can fund a superhedging portfolio on our claim. To find our optimality condition on $l$ and $r$ we will require strong duality, i.e. no duality gap and dual attainment, and we then use the form of the dual optimisers to give the form of the condition. Standard results do not give the form of the dual optimisers, so we will use PDE arguments to help suggest the form of the dual functions we need.

\section{Heuristic PDE arguments for duality} \label{heuristics}
The following analysis is motivated by \cite[Section~4]{Henry-Labordere:notes}. Suppose we want to superhedge the option with payoff $F(x,t)$, and to fit in with our LETF setup we assume that $X_t$ is a Geometric Brownian Motion (the argument is easily transferrable to the Brownian payoff). Initially we choose some region $\mathcal{D}$ (which will correspond to $\mathcal{R}^\complement$) and a function $\lambda(x)$ representing a static portfolio of call options at all strikes. 

We set our dynamic trading strategy to be
\begin{equation*}
\gamma(x,t):=\begin{cases}
F^\lambda(x,t) &\text{for} (x,t)\notin \mathcal{D}\\
\mathbb{E}^{(x,t)}[F^\lambda(X_{\tau_\mathcal{D}},\tau_\mathcal{D})]  &\text{for}  (x,t)\in \mathcal{D}\end{cases}
\end{equation*}
where $F^\lambda(x,t):=F(x,t)-\lambda(x)$. Then for $\gamma(x,t)+\lambda(x)$ to be a superhedge, we require
\begin{alignat}{2}
\label{supermart} \mathcal{L}\gamma:= \frac{x^2}{2}\partial_x^2\gamma + \partial_t\gamma &\leq 0 \quad && \forall (x,t)\\
\label{superhedge} \gamma &\geq F^\lambda \quad && \forall (x,t).
\end{alignat}
We can see immediately that \eqref{supermart} holds with equality in $\mathcal{D}$, and \eqref{superhedge} holds with equality in $\mathcal{D}^\complement$.

Consider a domain $\mathcal{D}$ which is the continuation region of a $K$-cave barrier, so for $(x,t)\in \mathcal{D}$ we have that $l(x):=\inf\{s<t:\medspace (x,s)\in\mathcal{D}\}$ and $r(x):=\sup\{s>t: \medspace (x,s)\in \mathcal{D}\}$ are independent of $t$. We want our superhedge to match the payoff on the boundary, so we require
\begin{equation}\label{boundaries}
\begin{split}
\gamma(x,l(x))&=F^\lambda(x,l(x))\\
\gamma(x,r(x))&=F^\lambda(x,r(x)).
\end{split}
\end{equation}
Then we wish to find $\mathcal{D}$, $\lambda$ such that
\begin{alignat*}{2}
\mathcal{L}\gamma&=0 \quad && \text{in} \medspace \mathcal{D}\\
\gamma&=F^\lambda \quad &&\text{on} \medspace \partial\mathcal{D}.
\end{alignat*}
Note that with this boundary condition we might expect a `smooth-fit' condition on the boundary, so that $\partial_t\gamma=\partial_t F^\lambda=\partial_t F$ on $\partial\mathcal{D}$, and then, writing $\eta=\partial_t\gamma$, we expect
\begin{alignat*}{2}
\mathcal{L}\eta&=0 \quad && \text{in} \medspace \mathcal{D}\\
\eta&=\partial_t F^\lambda=\partial_t F \quad && \text{on} \medspace \partial\mathcal{D}.
\end{alignat*}
We can then use Dynkin's Formula to deduce that 
\begin{align}
\eta(x,t)&=\mathbb{E}^{(x,t)}[\partial_t F(X_{\tau_\mathcal{D}},\tau_\mathcal{D})]=:M(x,t) \label{M} \\
\intertext{and so}
\gamma(x,t)&= -\int^{r(x)}_t M(x,v) \mathrm{d}v - \xi(x) \label{gammadef}
\end{align}
where $\xi(x)$ is some function, which we will choose to ensure $\mathcal{L} \gamma = 0$. We could take any upper limit, but we will see later that $r(x)$ is a natural choice.

With this form for the function $\gamma$, we can consider the boundary conditions, \eqref{boundaries}. Rearranging \eqref{boundaries}, we see that we must have
\begin{equation*}
  \lambda(x) = F(x,l(x)) - \gamma(x,l(x)) = F(x,r(x)) - \gamma(x,r(x)) 
\end{equation*}
Observing that $F(x,r(x)) =0$, we note that this holds whenever
\begin{equation}
 \label{dual}
 \Gamma(x) := F(x,l(x)) + \int_{l(x)}^{r(x)} M(x,v) \, \mathrm{d}v = 0 \quad \forall x \in \mathcal{D}.
\end{equation}


More generally, if we only require that \eqref{superhedge} holds, a necessary condition on the boundary is that 
\begin{equation*}
  \lambda(x) \ge \max\{F(x,l(x)) + \int_{l(x)}^{r(x)} M(x,v) \, \mathrm{d}v + \xi(x), \xi(x)\},
\end{equation*}
and we see that if $\Gamma(x) = 0$, it is sufficient to take $\xi(x) = \lambda(x)$. Since $\xi$ was chosen to make $\mathcal{L}\gamma=0$, this will effectively fix $\lambda$.  In the next section we will see that it is sufficient for \eqref{superhedge} to hold on the boundaries in order to deduce that it holds in the interior as well.

Then, to summarise this section, given a set $\mathcal{D}$ which is the continuation region of a $K$-cave barrier, we (heuristically) can construct functions $\gamma_{\mathcal{D}}$ (given by \eqref{gammadef}) and $\lambda_{\mathcal{D}}(x) := \max\{F(x,l(x)) + \int_{l(x)}^{r(x)} M(x,v) \, \mathrm{d}v + \xi(x), \xi(x)\}$ such that \eqref{supermart} and \eqref{superhedge} hold. If in addition $\tau$ is a (uniformly integrable) stopping time such that $X_{\tau} \sim \mu$, then:
\begin{align*}
\Ep{F(X_{\tau},\tau)}& \leq \Ep{\gamma_{\mathcal{D}}(X_{\tau},\tau) + \lambda_{\mathcal{D}}(X_{\tau})}\\
  & \leq \gamma_{\mathcal{D}}(X_0,0) + \int \lambda_{\mathcal{D}}(x) \mu(\di x),
\end{align*}
and therefore
\begin{equation}
\label{weakdual}
\sup_{\tau: X_{\tau} \sim \mu}\Ep{F(X_{\tau},\tau)}\leq \inf_{\mathcal{D}} \left\{\gamma_{\mathcal{D}}(X_0,0) + \int \lambda_{\mathcal{D}}(x) \mu(\di x)\right\}.
\end{equation}
Moreover, if $\mathcal{D}$ is such that $X_{\taud} \sim \mu$, $\gamma_{\mathcal{D}}(X_t \wedge \taud,t\wedge \taud)$ is a martingale, and $\Gamma(x) = 0$, the inequalities above are equalities for $\taud$, and so the supremum and the infimum coincide. 

Our aim in the next section will be to make these heuristic arguments rigorous whilst showing that, in fact, any set $\mathcal{D}$ which is the continuation region of a $K$-cave barrier embedding $\mu$ \emph{and} which satisfies \eqref{dual}  (or a slightly refined version of \eqref{dual}) gives equality in \eqref{weakdual}.

\section{Optimality} \label{optimality}
We have introduced the dual problem of choosing a $K$-cave barrier which embeds $\mu$ and such that $\Gamma(x)=0$. In this section we will make these heuristic arguments rigorous, and show that if we have a $K$-cave barrier that satisfies these conditions, then it does indeed give rise to an optimal embedding. We will modify the arguments presented in Cox and Wang \cite{Cox:2013ab}, using the heuristics of the previous section to motivate our choice of functions, and writing our problem for the Brownian motion, rather than GBM (although an essentially identical analysis holds for the GBM case). Hence, for a Brownian motion $B$, we wish to find an embedding $\tau$ of the form given in \thref{SGTheorem} (the hitting time of a $K$-cave barrier), and functions $G(x,t)$ and $H(x)$ such that
\begin{align}
&\bullet F(x,t)\leq G(x,t)+H(x) \text{ everywhere}   \label{D1} \\
&\bullet G(B_t,t) \text{ is a supermartingale} \label{D2} \\
&\bullet F(B_{\tau},\tau)=G(B_{\tau},\tau)+H(B_{\tau})  \label{D3} \\
&\bullet G(B_{t \wedge \tau},t \wedge \tau) \text{ is a martingale.} \label{D4}
\end{align}

We use the previous section to motivate a possible form of our super-replicating portfolio, and we will see that it is highly dependent on the region $\mathcal{D}$. The idea here is that the portfolio we propose, which depends heavily on the stopping region, is \textquoteleft dual feasible' for any stopping region, and then the correct choice of our region $\mathcal{D}$, or equivalently our curves $l,$ $r$, will correspond to satisfying the complementary slackness conditions of our primal-dual problem. The conditions \eqref{D1} and \eqref{D2} are our dual conditions, i.e. our dual problem is to minimise $\Ep{G(B_{\tau},\tau)+H(B_{\tau})}$ over functions $G,$ $H$ such that \eqref{D1}, \eqref{D2} hold. Then \eqref{D3} and \eqref{D4} are the complementary slackness conditions. In Section \ref{discrete} we prove that our choices of $G,$ $H$ are indeed the correct ones, so the condition we give is both necessary and sufficient. 

Consider a $K$-cave barrier $\mathcal{R}$ with continuation region $\mathcal{D}=\mathcal{R}^\complement$ and let $\tau_{\mathcal{D}}$ be the associated exit time, so $\tau_{\mathcal{D}}=\inf\{t\geq 0:t\notin (l(B_t),r(B_t))\}$. Then, looking at \eqref{M} and \eqref{gammadef}, we define
\begin{align*}
G(x,t)&:=G^*(x,t) -Z(x),\\
\text{where}\qquad
G^*(x,t)&:=-\int^{r(x)}_t M(x,s)\mathrm{d}s,\\
M(x,t)&:=\Eps{(x,t)}{\partial^-_t F\left(B_{\taud},\taud\right)}=-\frac{\beta^2}{2}\mathbb{E}^{(x,t)}\left[h(B_{\tau_{\mathcal{D}}},\tau_{\mathcal{D}})\mathbbm{1}\{\tau_{\mathcal{D}}\leq K(B_{\tau_{\mathcal{D}}})\}\right],
\end{align*}
and $Z(x)$ is chosen as above to ensure that $G(B_t,t)$ is a martingale in $\mathcal{D}$. Here we have taken the Brownian motion payoff, and the only difference if we take the LETF payoff is that $\frac{\beta^2}{2}$ becomes $\frac{\beta(\beta-1)}{2}$.

Since $h$ is a non-negative martingale, 
\begin{equation*}
M(x,t)-\frac{\beta^2}{2}\mathbb{E}^{(x,t)}\left[h(B_{\tau_{\mathcal{D}}},\tau_{\mathcal{D}})\mathbbm{1}\{\tau_{\mathcal{D}} > K(B_{\tau_{\mathcal{D}}})\}\right]=-\frac{\beta^2}{2}h(x,t)
\end{equation*}
and then we have
\begin{align*}
M(x,t)=-\frac{\beta^2}{2}h(x,t)\quad \text{for}\thickspace & (x,t)\in\{(x,t):t\leq l(x)\}\\
-\frac{\beta^2}{2}h(x,t)\leq M(x,t)\leq 0\quad \text{for}\thickspace & (x,t)\in D=\{(x,t):l(x)\leq t\leq r(x)\}\\
M(x,t)=0\quad \text{for}\thickspace & (x,t)\in\{(x,t):t\geq r(x)\}.
\end{align*}

Suppose that our $K$-cave barrier embeds $\mu_l$ to the left of $K$, i.e. along $l$, and $\mu_r$ to the right of $K$, along $r$. We say that both barriers are attainable at $x$ if $x\in \mathrm{supp}(\mu_l)\cap\mathrm{supp}(\mu_r)$. Define
\begin{equation*}
\Gamma(x):=F(x,l(x))+ \int^{r(x)}_{l(x)} M(x,v) \mathrm{d}v.
\end{equation*}
Then, from the heuristics in the previous section, we propose the following condition on our barriers $l$ and $r$ for optimality:
\begin{equation}
\label{Gamma} \tag{$\Gamma$}
\begin{aligned}
\Gamma(x)\geq0 \quad &\mu_l \text{-a.s.} \\
\Gamma(x)\leq0 \quad &\mu_r \text{-a.s.}
\end{aligned}
\end{equation}

\begin{theorem}
\thlabel{Opt}
If $\mathcal{R}$ is a $K$-cave barrier that embeds a distribution $\mu$ and also satisfies \eqref{Gamma}, then $\tau_\mathcal{D}$ is optimal.
\end{theorem}

To show this we first need to show that our function $G^*$ is such that we can choose $Z$ and $H$ to give the required properties. First, let $x^*:=\inf\{x: \medspace l(x)=K(x)=r(x) \}$, where we set $\inf \emptyset = \infty$ if our barriers never meet. Note that if $x^*<\infty$, then our distribution $\mu$ embeds no mass above $x^*$ and so any pair of barriers embedding $\mu$ must meet at $x^*$, and in particular our process is always stopped below this point, or before $H_{x^*}=\inf\{t\geq 0: \medspace B_t=x^*\}$.
\begin{lemma}
\thlabel{Glemma}
We can find a function $Z$ such that the process
\begin{equation*}
G(B_{t\land \taud},t\land \taud) \qquad \text{is a martingale,}
\end{equation*}
and
\begin{equation*}
G(B_t,t) \qquad \text{is a supermartingale up to $H_{x^*}$.}
\end{equation*}
\end{lemma}
\begin{proof}

We first show that we can find an increasing process $A_t=A(B_t)$, depending only on $B_t$, such that $G^*(B_t,t)-A_t$ is a martingale in $\mathcal{D}$, and a supermartingale in general. We note that, for either of our payoffs, $h(x,t)<\infty$ for any $(x,t)$ and $h$ is integrable on bounded domains. This means that $|G^*|$ is bounded on compact sets in space for all $t\geq0$, and so all of the terms in the following arguments are well defined. In much of what follows we will take our process at some point $(B_t,t)$ and consider letting it run until some stopping time, perhaps $\tau=\inf \{u>0: \medspace |B_{t+u}-B_t|\geq\delta\} \wedge \epsilon$ for some small $\delta$ and $\epsilon$. 

1. \textit{Show} $G^*(B_t,t)$ \textit{is a submartingale in }$\mathcal{D}$: First take $(B_t,t)\in \mathcal{D}$, and $\tau$ a stopping time of the above form such that $t+\tau<\tau_{\mathcal{D}}$, so we remain in the continuation region. Then,
\begin{align*}
\mathbb{E}[G^*(B_{t+\tau},t+\tau)-&G^*(B_t,t)|\mathcal{F}_t]\\
&=\mathbb{E}\left[-\int^{r(B_{t+\tau})}_{t+\tau} M(B_{t+\tau},u)\mathrm{d}u + \int^{r(B_t)}_t M(B_t,u)\mathrm{d}u \Big| \mathcal{F}_t\right]\\
&=\mathbb{E}\left[-\int^{r(B_t)}_t \Big( M(B_{t+\tau},u+\tau)-M(B_t,u) \Big) \mathrm{d}u \Big| \mathcal{F}_t\right]\\
& \qquad \qquad \qquad \qquad \qquad +\mathbb{E}\left[\int^{r(B_t)+\tau}_{r(B_{t+\tau})} M(B_{t+\tau},u) \mathrm{d}u \Big| \mathcal{F}_t\right].
\end{align*}
It is natural to split the integrals up in this way since we know that in the continuation region $M$ is a martingale, and so we hope to use Fubini and the martingale property to argue that the first term is zero. However, our Brownian motion does not stay within $\D$ for all $u\in(t,r(B_t))$, as shown in Figure~\ref{timeindep}, and so we cannot use this martingale property and instead must argue about the sign of this term.
\begin{figure}[h]
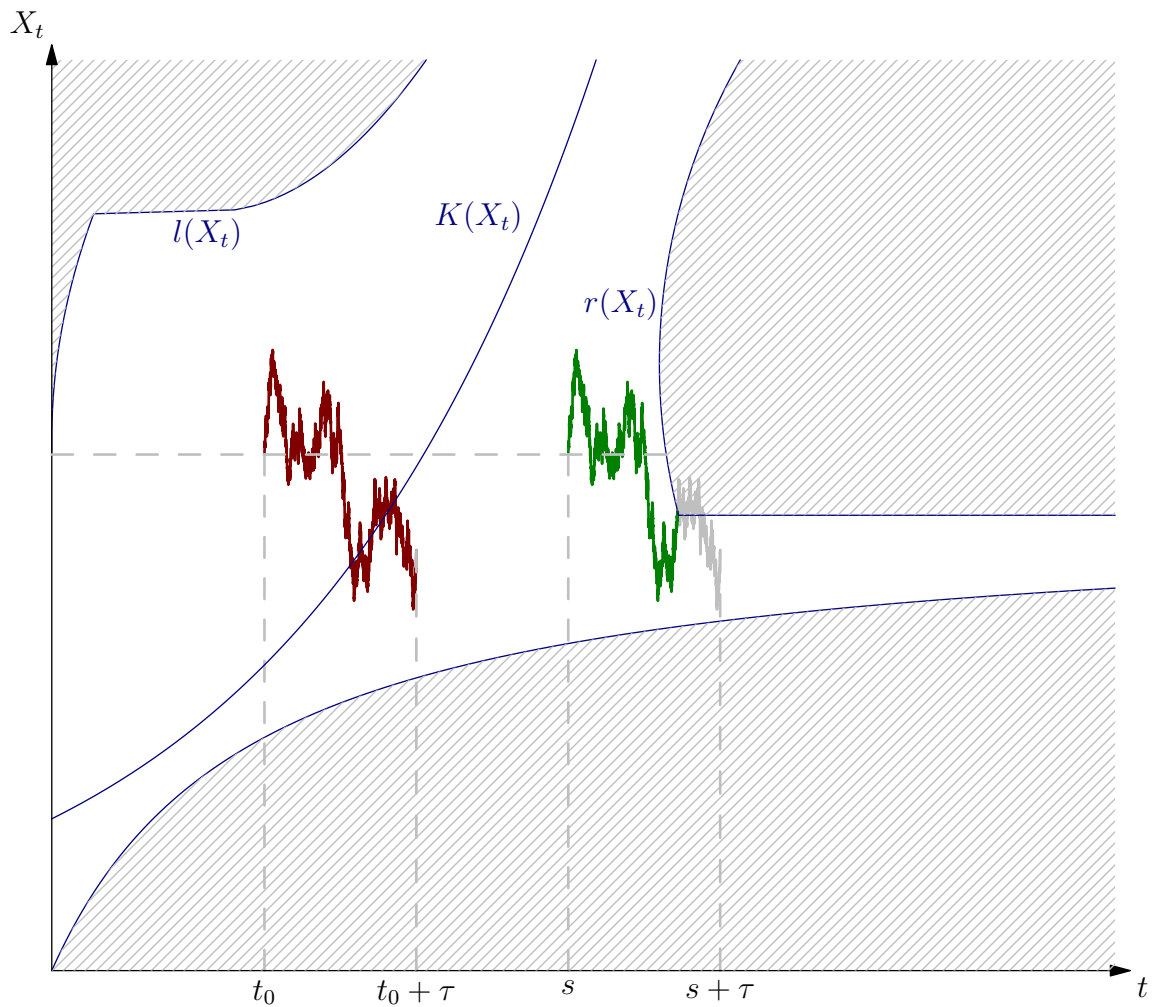

\centering
\begin{asy}[width=\textwidth]
    import graph;
    import stats;
    import patterns;
    
     // We construct a Brownian motion of time length T, with N time-steps
    int N = 3000;
    //int N = 300;
    real T = 3.5;
    real A = 0.5;
    real dt = A/N;
    real B0 = 1.7;
    real t0 = 1.7;
    real s = 1;
    
    real xmax = 3;
    real xmin = 0;
    real r1 = 1.5;
    real l1 = 2.5;
    real x1 = 0.4;
    real x2 = 0.2;
    real eps = 0.05;
    real eps1 = 0.05;
    real kint = 0.5;

    real[] B; // Brownian motion
    real[] t; // Time

    path BM;
    path BMs;
    
    // Seed the random number generator. Delete for a "random" path:
    srand(4);

    B[0] = B0;
    t[0] = t0;

    BM = (t[0],B[0]);
    BMs = (t[0]-s,B[0]);
 // Define a barrier 

    real R1(real y) {return 2*exp(((y-2)**2)/8);}
    real R2(real y) {return (1/(r1-y)-1/r1);}
    real R2inv(real z) {return r1 - 1/(z+1/r1) ;}
    real R(real y) {return (y>r1) ? R1(y) : R2(y);}
    real L1(real y) {return  (y>B0+eps) ? ((y<l1) ? ((y-B0-eps)**2)/4 : sqrt(y-l1)+((1.45^2)/4)) : 0 ;}
    real L(real y) {return (y>B0) ? L1(y) : 0 ;}
    real K(real y) {return (y>0) ? log(y)-log(kint) : 0 ;}

    int H = N+1;
    int H2 = N+1;
    int BMstop;
    int BMstop2 = N+2;
    
      for (int i=1; i<N+1; ++i)
    {
      B[i] = B[i-1] + Gaussrand()*sqrt(dt);
      t[i] = t0+i*dt;
      BM = BM--(t[i],B[i]);
      BMs = BMs--(t[i]-s,B[i]);
    
          if ((H==N+1)&&(t[i]>=R(B[i])))
      {
	H = i;
	BMstop = length(BM);
      }

      if ((H==N+1)&&(t[i]<L(B[i])))
      {
	H = i;
	BMstop = length(BM);
      }
      if ((H2==N+1) && (B[i]<0))
      {
	B[i] = 0;
	BMstop2 = length(BM);
	H2 = i;
      }	
      if ((H2==N+1) && (B[i]>=xmax))
      {
	B[i] = xmax;
	BMstop2 = length(BM);
	H2 = i;
      }			
    }
    
    if (H==N+1)
    BMstop = length(BM);

    pen p = deepgreen + 1;
    pen p2 = mediumgray + 1;
    pen p3 = brown + 1;

    if (H<N+1)
    draw(subpath(BM,BMstop,BMstop2),p2);
    draw(subpath(BMs,BMstop,BMstop2),p3);

    draw(subpath(BM,1,BMstop),p);
    draw(subpath(BMs,1,BMstop),p3);

    pair tau = point(BM,BMstop2);
    pair tau2 = point(BMs,BMstop2);

    draw((t0,0)--(t0,B0),p2+dashed);
    label("$s$",(t0,0),S);
    draw((t0+A,0)--tau,p2+dashed);
    label("$s+\tau$",(t0+A,0),S);
    draw((t0-s,0)--(t0-s,B0),p2+dashed);
    label("$t_0$",(t0-s,0),S);
    draw((t0+A-s,0)--tau2,p2+dashed);
    label("$t_0+\tau$",(t0+A-s,0),S);

    pen q = black + 0.5;
    
    draw("$l(X_t)$",graph(L1,identity,B0,xmax),SE,deepblue+0.5);
    // draw((L1(r1),r1)--(T,r1),NW,deepblue+0.5);
    draw(graph(R2,identity,xmin,R2inv(T)),NW,deepblue+0.5);
    draw("$r(X_t)$",graph(R1,identity,r1,xmax),SW,deepblue+0.5);
    draw((R1(r1),r1)--(T,r1),NW,deepblue+0.5);
    draw(L=Label("$K(X_t)$",Relative(0.8) ,NW),graph(K,identity,kint,xmax),deepblue + 0.5);

    path barrier1 = (graph(L1,identity,xmax,B0+eps)--(0,xmax)--(L(xmax),xmax)--cycle);

    add("hatch",hatch(1mm,mediumgray));

    fill(barrier1,pattern("hatch"));

    draw((0,xmin)--(0,xmax+eps1),q,Arrow);
    draw((0,xmin)--((T+eps1),xmin),q,Arrow);
    label("$t$",(T+eps1,xmin),SE);
    label("$X_t$",(0,(xmax+eps1)),NW);
    
    draw((L1(B0),B0)--(R1(B0),B0),p2+dashed);
    
     path barrier3 = (graph(R1,identity,xmax,r1)--(R1(r1),r1)--(T,r1)--(T,xmax)--cycle);

    path barrier4 = (graph(R2,identity,xmin,R2inv(T))--(T,xmin)--cycle);

    add("hatch",hatch(1mm,mediumgray)); 
    
    fill(barrier3,pattern("hatch"));

    fill(barrier4,pattern("hatch"));

    \end{asy}
\caption{Here we have a path leaving from $(B_{t_0},t_0)$ running for a time $\tau$ inside $\D$ and we consider moving this path along the time axis, so we may now exit $\D$.}
\label{timeindep}
\end{figure}

Note that $M(B_{t+\tau},u)=0$ if $u\geq r(B_{t+\tau})$, and $M(B_{t+\tau},u)\leq 0$ otherwise, so the final term of the above is non-negative. We can also show that the other term in the final expression is non-negative. Let $\taud^{(x,t)}=\inf\{s\geq 0: \medspace (x+B_s,t+s)\notin \D\}=\inf\{s\geq 0: \medspace u+s\geq r(B_{t+s})\}$ be the hitting time of the boundary after we set off from $(x,t)$. Take $u\in (t,r(B_t))$ and let $\hat{\tau}_{\mathcal{D}}=\taud^{(B_t,u)} \wedge \tau$. When $\hat{\tau}_{\mathcal{D}}=\tau$ we have $M(B_{t+\tau},u+\tau)=M(B_{t+\hat{\tau}_{\mathcal{D}}},u+\hat{\tau}_{\mathcal{D}})\leq 0$, and when $\hat{\tau}_{\mathcal{D}}<\tau$ we have that $M(B_{t+\tau},u+\tau)\leq 0 =M(B_{t+\hat{\tau}_{\mathcal{D}}},u+\hat{\tau}_{\mathcal{D}})$. Therefore,
\begin{equation*}
\mathbb{E}\left[M(B_{t+\tau},u+\tau)|\mathcal{F}_t\right] \leq \mathbb{E}\left[M(B_{t+\hat{\tau}_{\mathcal{D}}},u+\hat{\tau}_{\mathcal{D}})|\mathcal{F}_t\right]=M(B_t,u)
\end{equation*}
since $M(B_t,t)$ is a martingale in $\mathcal{D}$. Swapping the expectation and the integral by Tonelli's theorem, we conclude that
\begin{equation}
\label{Mleq}
\mathbb{E}\left[-\int^{r(B_t)}_t \Big( M(B_{t+\tau},u+\tau)-M(B_t,u) \Big) \mathrm{d}u \Big| \mathcal{F}_t\right]\geq 0.
\end{equation}

Provided we have integrability, this tells us that $G^*(B_t,t)$ is a submartingale in $\mathcal{D}$, and therefore the Doob-Meyer Decomposition Theorem tells us that there exists a unique, increasing, predictable process $A_t$ such that $M_t=G^*(B_t,t)-A_t$ is a martingale in $\mathcal{D}$. But,
\begin{equation*}
\Ep{|G^*(B_t,t)|}\leq \Ep{-\int^{r(B_t)}_0 M(B_t,s) \di s}\leq\Ep{F(B_t,0)}<\infty \quad \forall t
\end{equation*}
for either of our payoffs, and so we have integrability.

2. $A_t$ \textit{depends only on} $B_t$: To think more about $A_t$ we consider, as usual, a time $t<\taud$ and then run our process from $t$ up until a small stopping time $\tau$ such that $t+\tau<\taud$, but now we imagine moving this path along the time axis. We then have $t<\taud$, $\tau<\taud^{(B_t,t)}$ and we take $s<r(B_t)$ such that $\tau<\taud^{(B_t,s)}$. Then,
\begin{align} 
\label{Gts}
\Ep{G^*(B_{t+\tau},s+\tau)-G^*(B_t,s) | \F_t}&=\Ep{G^*(B_{t+\tau},t+\tau)-G^*(B_t,t) | \F_t}\\
&\qquad +\Ep{\int^s_t\left(M(B_{t+\tau},u+\tau)-M(B_t,u)\right) \di u \Big| \F_t}.\notag
\end{align}
Since $s<s+\tau<\taud$ and $t<t+\tau<\taud$, and by the shape of our boundaries, we have that $(B_t,u),(B_{t+\tau},u+\tau) \in \D$ for $u\in(t,s)$, and as $M(B_t,t)$ is a martingale in $\D$, we have that 
\begin{equation*}
\Ep{M(B_{t+\tau},u+\tau) | \F_t}=\Ep{M(B_t,u) | \F_t}=M(B_t,u)
\end{equation*}
for all $u\in(t,s)$. By Fubini the final term in \eqref{Gts} is $0$, so
\begin{equation}
\label{tseq}
\Ep{G^*(B_{t+\tau},s+\tau)-G^*(B_t,s) | \F_t}=\Ep{G^*(B_{t+\tau},t+\tau)-G^*(B_t,t) | \F_t}.
\end{equation}
This tells us that in $\D$, $A_t$ depends only on $B_t$ and not directly on $t$. If we now consider taking any $s$, but keeping $t$ such that $t+\tau<\taud$, then we still have \eqref{Gts}, but now we can show that the final term is actually non-positive.

Since we now consider any $s$, we will no longer always be in the continuation region, and we need to consider crossing the boundaries. We know from \thref{SGTheorem} that our right-hand boundary $r$ is a barrier, and $l$ is an inverse barrier. If we have $\tau<\taud^{(B_t,t)}$, then $(B_{t+u},t+u)\in\D$ for every $u\in (0,\tau)$, so in particular we do not cross the left hand boundary $l$. If $t<s$ then, since $l$ is an inverse barrier, we must also have that $s+u>l(B_{t+u})$ for every $u\in (0,\tau)$, so shifting this part of our path to the right cannot cause us to cross $l$. We can however cross $r$, so we need to argue exactly as with \eqref{Mleq} to see that
\begin{equation*}
\Ep{\int^s_t\left(M(B_{t+\tau},u+\tau)-M(B_t,u)\right) \di u \Big| \F_t}\leq0
\end{equation*}
and so
\begin{equation*}
\Ep{G^*(B_{t+\tau},s+\tau)-G^*(B_t,s) | \F_t}\leq\Ep{G^*(B_{t+\tau},t+\tau)-G^*(B_t,t) | \F_t}.
\end{equation*}

If we take $s<t$ then we instead have that $s+u<r(B_{t+u})$ for every $u\in (0,\tau)$, and so we do not cross $r$ but could cross $l$. The argument here is similar in that we let $\tilde{\tau}_{\mathcal{D}}=\taud^{(B_t,u)} \wedge \tau$, take $u\in(s,t)$ and compare $M(B_{t+\tau},u+\tau)$ and $M(B_{t+\tilde{\tau}_\D},u+\tilde{\tau}_\D)$. On $\{\tilde{\tau}_\D=\tau\}$ we clearly have $M(B_{t+\tau},u+\tau)=M(B_{t+\tilde{\tau}_\D},u+\tilde{\tau}_\D)$, but when $\tilde{\tau}_\D<\tau$ we have
\begin{align*}
\Ep{M(B_{t+\tilde{\tau}_\D},u+\tilde{\tau}_\D) | \F_t}&=\Ep{\frac{-\beta^2}{2}h(B_{t+\tilde{\tau}_\D},u+\tilde{\tau}_\D) \Big| \F_t}\\
&= \Ep{\frac{-\beta^2}{2}h(B_{t+\tau},u+\tau) \Big| \F_t}\\
&\leq \Ep{M(B_{t+\tau},u+\tau) | \F_t}
\end{align*}
by the Optional Sampling Theorem, since both our stopping times are bounded. Combining these as before and using Fubini, we again have
\begin{equation*}
\Ep{\int^s_t\left(M(B_{t+\tau},u+\tau)-M(B_t,u)\right) \di u \Big| \F_t}\leq0
\end{equation*}
and so for $t<t+\tau<\taud$ and any $s$, we have that
\begin{equation}
\label{tsleq}
\Ep{G^*(B_{t+\tau},s+\tau)-G^*(B_t,s) | \F_t}\leq\Ep{G^*(B_{t+\tau},t+\tau)-G^*(B_t,t) | \F_t}=\Ep{A_{t+\tau}-A_t |\F_t}.
\end{equation}

3. $G(B_t,t)$ \textit{has the desired properties}: We now combine the above two results to show that we have the supermartingale property we require, noting that we already have the martingale property in $\D$ as this is how we chose $A$. Consider now arbitrary $s$ and $\tau$ and suppose that we can fix a $t$ such that $(s,B_t)\in\D$ and $\tau<\taud^{(B_t,s)}$. Then from \eqref{tseq} and \eqref{tsleq} we have
\begin{equation*}
\Ep{G^*(B_{t+\tau},t+\tau)-G^*(B_t,t) | \F_t}\leq\Ep{G^*(B_{t+\tau},s+\tau)-G^*(B_t,s) | \F_t}.
\end{equation*}
We can use this to give the following:
\begin{align*}
\Ep{G^*(B_{t+\tau},t+\tau)-A_{t+\tau} | \F_t}&\leq\Ep{G^*(B_{t+\tau},s+\tau)-A_{t+\tau} | \F_t}\\ 
& \qquad \qquad \qquad \qquad +G^*(B_t,t)-G^*(B_t,s)\\
&=G^*(B_t,t) + \Ep{G^*(B_{t+\tau},s+\tau)-G^*(B_t,s) | \F_t}\\
& \qquad \qquad \qquad \qquad -\Ep{A_{t+\tau} | \F_t}\\
&= G^*(B_t,t) + \Ep{A_{t+\tau}-A_t | \F_t} -\Ep{A_{t+\tau} | \F_t}\\
&=G^*(B_t,t)-A_t,
\end{align*}
which is exactly the supermartingale property we are looking for. 

It will not always be the case that we can find such a $t$ as above, in fact for a given $t$, $\tau$ we may find that $\tau>\taud^{(B_t,s)} \medspace \forall s$ such that $(B_t,s)\in\D$. We then need to find a sequence of stopping times that sum to $\tau$ and use the above on each of the intervals. Suppose first that our curves $l$, $r$ do not meet, or they do so well away from $t$ and $t+\tau$. We can then choose some $s\in(l(B_t),r(B_t))$ (we will usually take $s=K(B_t)$ for simplicity unless we have $l(B_t)=K(B_t)$ or $r(B_t)=K(B_t)$) and we run the process from $(B_t,s)$ until we hit a boundary, call this stopping time $\sigma_1$. We then move back into our continuation region and set off from $(K(B_{s+\sigma_1}),B_{s+\sigma_1})$, and run again for a time $\sigma_2$ until we hit the boundary. Provided our barriers do not meet we can continue this until we reach $s+\tau$ in a finite number of steps. We can then write $\Ep{G^*(B_{t+\tau},t+\tau)-G^*(B_t,t) | \F_t}$ as a telescoping sum and show the inequality as before. From the exact argument above with $\tau$ when we do not leave the region, we have that
\begin{equation*}
\Ep{G^*(B_{t+\sigma_1},t+\sigma_1)-A_{t+\sigma_1 }| \F_t}\leq G^*(B_t,t)-A_t,
\end{equation*}
and also
\begin{equation*}
\Ep{G^*(B_{t+\sigma_{j+1}},t+\sigma_{j+1})-A_{t+\sigma_{j+1} }| \F_t}\leq\Ep{G^*(B_{t+\sigma_j},t+\sigma_j)-A_{t+\sigma_j }| \F_t}
\end{equation*}
for our stopping times $\{\sigma_j\}_j$ where $\sigma_j=\tau$ for some $j$. We then combine these results in our telescoping sum to get the supermartingale property as before. If $B_{t+\tau}<x^*$ then we can always find a finite sequence of stopping times that sum to $\tau$. The only other case is where $B_{t+\tau}=x^*$. In this case we again require a sequence of stopping times, but this time we will could have infinitely many, with the sum converging to $\tau$, but then we can work as before but using Fubini to interchange our expectation and the infinite sum.

We now know that we can find an increasing process $A_t$, dependent only on $B_t$, such that $G^*(B_t,t)-A_t$ is a martingale up until $\taud$ and a supermartingale in general. We know (\cite[Chapter~X, Section~2]{Revuz:1999aa}) that any continuous additive functional $A_t$ of linear Brownian Motion can be written as
\begin{equation} \label{CAF}
A_t=f(B_t)-f(B_0)-\int^t_0 f'_-(B_s)\di B_s
\end{equation}
for some convex function $f$. Then we must have that for any $s,t$,
\begin{equation*} 
\Ep{A_t-A_s | \F_t}=\Ep{f(B_t)-f(B_s) | \F_t}.
\end{equation*}
We therefore choose $Z(x)=f(x)$ to give the result.
\end{proof}

We now return to proving \thref{Opt} by choosing the function $H$.
\begin{proof}[Proof of \thref{Opt}]
Our choice of $H$ should be to give $F=G+H$ on the boundaries, and $F\leq G+H$ in general. We have
\begin{equation*}
G(x,t)+Z(x)=G^*(x,t)=-\int_t^{r(x)} M(x,s) \di s,
\end{equation*}
so for any $x,t$ 
\begin{align*}
t<K(x) \quad & \implies \quad F_t(x,t)=-\frac{\beta^2}{2}h(x,t)\leq M(x,t)=G^*_t(x,t),\\
t>K(x) \quad & \implies \quad F_t(x,t)=0\geq M(x,t)=G^*_t(x,t).
\end{align*}
From these derivatives we can see that if $G(x,l(x))+H(x)\geq F(x,l(x))$ and $G(x,r(x))+H(x)\geq F(x,r(x))$ (where $l(x),r(x)$ are possibly $0,\infty$ respectively), then $G(x,t)+H(x)\geq F(x,t)$ everywhere, as required. 

Let $H(x)=Z(x)+\left(\Gamma(x)\right)^+$, so $G(x,t)+H(x)=G^*(x,t)+\left(\Gamma(x)\right)^+$. This is a pathwise superhedging strategy since
\begin{align*}
\Gamma(x)>0 &\implies 
\begin{cases}
G(x,l(x))+H(x)=F(x,l(x)) \\
G(x,r(x))+H(x)=\Gamma(x)>F(x,r(x)),
\end{cases} \\
\Gamma(x)<0 &\implies 
\begin{cases}
G(x,l(x))+H(x)=F(x,l(x))-\Gamma(x)>F(x,l(x)) \\
G(x,r(x))+H(x)=F(x,r(x)),
\end{cases} \\
\Gamma(x)=0 &\implies
\begin{cases}
G(x,l(x))+H(x)=F(x,l(x))\\
G(x,r(x))+H(x)=F(x,r(x)).
\end{cases}
\end{align*}

For $x\in\mathrm{supp}(\mu_r)$ we require $G(x,r(x))+H(x)=F(x,r(x))$, which holds by the above when $\Gamma(x)\leq 0$. Similarly, for $x\in\mathrm{supp}(\mu_l)$ we have $G(x,l(x))+H(x)=F(x,l(x))$ when $\Gamma(x)\geq 0$. Also note that for $x\notin\mathrm{supp}(\mu_l)\cup\mathrm{supp}(\mu_r)$ we can choose any $H(x)$ that gives the superhedging property.

We now have the desired properties for $G$ and $H$ and prove our theorem as follows. Let $\tau'$ be any other stopping time that embeds $\mu$. Then,
\begin{align*}
\mathbb{E}[F(B_{\tau_\mathcal{D}},\tau_\mathcal{D})]&=\mathbb{E}[G(B_{\tau_\mathcal{D}},\tau_\mathcal{D})]+\mathbb{E}[H(B_{\tau_\mathcal{D}})]\\
&=G(B_0,0)+\int_\mathbb{R} H(x) \mu(\mathrm{d}x)\\
&\geq \mathbb{E}[G(B_{\tau'},\tau')]+\int_\mathbb{R} H(x) \mu(\mathrm{d}x)\\
&\geq \mathbb{E}[F(B_{\tau'},\tau')].
\end{align*}
The first equality follows from our assumption \eqref{Gamma}, so, as we have shown above, our processes $G(B_t,t)+H(B_t)$ and $F(B_t,t)$ agree on the boundary. Also note that $\Ep{H(B_{\tau_\mathcal{D}})}<\infty$ since $A_{\tau_\mathcal{D}}$ is integrable. In the second line we use the martingale property of $G(B_t,t)$ in $\mathcal{D}$ and rewrite the $H$ term as an integral to make it clear that this term does not change, since both stopping times embed $\mu$. The inequality then follows since $G(B_t,t)$ is a supermartingale up to $H_{x^*}$ and we know that for any embedding $\tau'$ of $\mu$ we have that $B_{\tau'}\leq x^*$. The final inequality is true since we have shown above that $G+H\geq F$ everywhere.
\end{proof}

\begin{remark}
The case for geometric Brownian motion is similar, noting that the measure associated with a continuous additive functional of a geometric Brownian motion is a Radon measure, and therefore we again have the representation \eqref{CAF} (see \cite[Chapter~X, Section~2]{Revuz:1999aa}). 
\end{remark}

\begin{remark}
We can see immediately that if $\mu$ allows us to choose $r=K$ as our right-hand barrier, then the condition $\Gamma=0$ is always satisfied, since $\partial_t^-F(B_t,t)$ is a martingale before crossing $K$, and so this is the optimal choice.
\end{remark}

\section{Necessity of Condition $(\Gamma)$ via Linear Programming} \label{discrete}
Our aim now is to show the converse of \thref{Opt}, that is if we have a $K$-cave barrier that does not satisfy \eqref{Gamma}, then it does not give the optimal embedding. To do this we show that the functions $G$, $H$ we have chosen are the correct choice of the functions in our \textquoteleft dual' problem of finding the cheapest superhedging portfolio. We have proposed one feasible superhedging portfolio, and this portfolio gives the sufficient condition \eqref{Gamma}, but other feasible dual formulations could give different conditions, so we show that our condition is also necessary. To show this we require some form of strong duality result, which furthermore gives the form of the dual optimisers. To the best of our knowledge these results are not available in our current setup, but we can discretise our problem and then use standard results from infinite-dimensional linear programming. 

In \cite{Cox:2016ab}, we consider discretising an optimal Skorokhod embedding problem to create an optimal stopping problem for a random walk, which can then be considered as a linear programming problem. This problem has a well-defined Fenchel dual and we are able to prove a strong duality result in this discrete setting. We also show that as we let the step size of our random walk shrink to zero, we can recover the optimal continuous time solution in certain cases. In particular, if we are maximising the expected value of a convex or concave function of our stopping time, then we recover the Rost or Root embeddings respectively. In \cite{Beiglboeck:2013aa} the authors introduce the cave embedding solution to the Skorokhod embedding problem, which can be seen as the combination of a Root and a Rost barrier, as is the case with the LETF problem. We show in \cite{Cox:2016ab} that we reproduce this cave embedding result also, and here we argue that the $K$-cave barriers can be done similarly. 

Suppose now that our target measure $\mu$ is bounded, with $x^*$ the smallest $x$ such that $\mu((x,\infty))=0$, and $x_*$ the largest $x$ such that $\mu((-\infty,x_*))=0$. We work on the grid $\left(x^N_j,t_n^N\right)=\left(\frac{j}{\sqrt{N}}, \frac{n}{N} \right)$ for $j\in\{\lfloor x_*\sqrt{N}\rfloor,\lfloor x_*\sqrt{N}\rfloor+1,\ldots,\lfloor x^*\sqrt{N}\rfloor\}=:\mathcal{J}$ and $n\geq 0$. Let $j^N_0:=\lfloor x_*\sqrt{N}\rfloor$, $j^N_1:=\lfloor x_*\sqrt{N}\rfloor+1$, $\ldots$, $j^N_L:=\lfloor x^*\sqrt{N} \rfloor$, where $L\sim\sqrt{N}$, so $\mathcal{J}=\{j^N_0,j^N_1,\ldots,j^N_L\}$. We also define $\mathcal{J}'=\{j^N_1,\ldots,j^N_{L-1}\}$, and $\mathcal{J}''=\{j^N_2, \ldots,j^N_{L-2}\}$. For each $N$ we choose $j^{*,N}\in\mathcal{J}$ so that $x_{j^*}^N=\frac{j^{*,N}}{\sqrt{N}} \rightarrow 0$ as $N\rightarrow \infty$. If $Y^N$ is the SSRW on this grid, started at $x_{j^*}^N$, then by Donsker's Theorem, $Y^N_{\lfloor Nt\rfloor}$ converges in distribution to a Brownian motion started at $0$. In the case of geometric Brownian motion we take $x_j^N=\e^{\frac{j}{\sqrt{N}}}$.

We also need a discretised version of our payoff $F$, say $\bar{F}^N$, chosen so that $\bar{F}^N(\lfloor \sqrt{N}x\rfloor,\lfloor Nt\rfloor)\rightarrow F(x,t)$ everywhere. In the Brownian setup our continuous-time payoff function is $F(x,t)=\\ \left(\mathrm{e}^{\beta x} \mathrm{e}^{-\frac{\beta^2}{2}t} -k\right)^+=(h(x,t)-k)^+$, where $h(X_t,t)$ is a martingale, and we write the discretised version with a similar martingale term. We have
\begin{align*}
\Ep{\exp(\beta Y^N_{n+1}) | Y^N_n} &= \exp(\beta Y^N_n) \left(\frac{1}{2}\exp\left(\frac{\beta}{\sqrt{N}}\right) + \frac{1}{2}\exp\left(-\frac{\beta}{\sqrt{N}}\right)\right) \\
&=\exp(\beta Y^N_n) \cosh\left(\frac{\beta}{\sqrt{N}}\right),
\end{align*}
and so $\bar{F}^N(j,t)=\bar{F}^N_{j,t}=\left(\mathrm{e}^{\beta x^N_j} \left(\cosh\left(\frac{\beta}{\sqrt{N}}\right)\right)^{-t} -k\right)^+$ has the same form as before. Note now that $\bar{F}_{j,n}^N\approx F(x^N_j,n \Delta t)=F(\frac{j}{\sqrt{N}},\frac{n}{N})$, or $F(x,t)\approx \bar{F}^N_{\lfloor x \sqrt{N} \rfloor, \lfloor  tN \rfloor}$, since $\left(\cosh\left(\frac{\beta}{\sqrt{N}}\right)\right)^{-Nt} \rightarrow \mathrm{e}^{-\frac{\beta^2}{2} t}$, as $N \rightarrow\infty$. In the case of Geometric Brownian motion the arguments are the same.

If $\ttau$ is a stopping time of our random walk $Y^N$, we can define the probabilities 
\begin{align*}
\pjt^N&=\P\left(Y_t^N=x^N_j, \medspace \tilde{\tau}\geq t+1 \right)\\
\qjt^N&=\P\left(Y_t^N=x^N_j, \medspace \tilde{\tau}=t\right).
\end{align*}
We can optimise over these using the one-to-one correspondence between (randomised) stopping times $\ttau$ and the probabilities $p,q$. In \cite{Cox:2016ab} we give the following primal-dual pair of problems:
\begin{align*}
\mathcal{P}^N: \medspace \sup_{p} & \Bigg\{ \sum_{\substack{j\in\mathcal{J}'' \\ t\geq 2}} \bar{F}_{j,t}^N \left(\half\left(p_{j-1,t-1}+p_{j+1,t-1}\right)-p_{j,t}\right) + \sum_{t\geq2} \bar{F}_{j^N_L,t}^N \half p_{j^N_{L-1},t-1} +  \sum_{t\geq2} \bar{F}_{j^N_0,t}^N \half p_{j^N_1,t-1}  \\
& \qquad \qquad + \sum_{t\geq2} \bar{F}_{j^N_{L-1},t}^N\left( \half p_{j^N_{L-2},t-1} - p_{j^N_{L-1},t}\right) + \sum_{t\geq2} \bar{F}_{j^N_1,t}^N\left( \half p_{j^N_2,t-1} - p_{j^N_1,t}\right)  \\
& \qquad \qquad \qquad \qquad \qquad \qquad \qquad +  \bar{F}_{j^*+1,1}^N \left(\half - p_{j^*+1,1}\right) + \bar{F}_{j^*-1,1}^N \left(\half - p_{j^*-1,1}\right) \Bigg\},  \\
\end{align*}
over $(\pjt)_{\substack{j\in\mathcal{J}' \\ t\geq 1}}$ subject to
\begin{alignat*}{2}
& \bullet \medspace (p_{j,t})\in l^1 && \\
& \bullet \medspace p_{j,t}\geq 0, \quad \forall j,t && \\
& \bullet \medspace \mathbbm{1}\{j=j^*\}+\sum_{t=1}^\infty \pjt \leq \sqrt{N}\left(\sum_i |x^N_i-x^N_j| \mu^N(\{x^N_i\}) - |x^N_{j^*}-x^N_j|\right)=:U^N_j, \quad &&\forall j\in\mathcal{J}' \\
& \bullet \medspace \pjt \leq \frac{1}{2}(p_{j-1,t-1} + p_{j+1,t-1}), \quad &&\forall t\geq 2, j\in\mathcal{J}'' \\
& \bullet \medspace p_{j^N_1,t}\leq\frac{1}{2} p_{j^N_2,t-1}, \quad &&\forall t\geq 2 \\
& \bullet \medspace p_{j^N_{L-1},t}\leq\frac{1}{2} p_{j^N_{L-2},t-1}, \quad &&\forall t\geq2  \\
& \bullet \medspace p_{j^*+1,1}\leq\frac{1}{2}, \quad p_{j^*-1,1}\leq\frac{1}{2} \\
& \bullet \medspace p_{j,1}=0, \quad &&\forall j\neq j^*\pm1. \\
\end{alignat*}
Which has dual problem
\begin{flalign}
\notag &\mathcal{D}^N: \medspace \inf_{\eta,\nu} \Bigg\{ \sum_{j\in\mathcal{J}'} \nu_j U_j + \half\left(\eta_{j^*+1,1}+\eta_{j^*-1,1}\right) +\half\left(\bar{F}^N_{j^*+1,1}+\bar{F}^N_{j^*-1,1}\right) \Bigg\}&\\
\notag &\text{over } (\nu_j)_{j\in\mathcal{J}'}, (\eta_{j,t})_{\substack{j\in\mathcal{J} \\ t\geq 1}} \text{ subject to } &\\
\notag &\bullet (\nu,\eta)\in l^{\infty} &\\ 
\label{Fdiscretesuper} &\bullet \medspace \eta_{j,t},\nu_j\geq 0, \quad &&\forall j,t &\\
\label{Fetasuper}  &\bullet \medspace \half\left(\eta_{j+1,t+1}+\eta_{j-1,t+1}\right)-\eta_{j,t}-\nu_j\leq \bar{F}^N_{j,t}-\half\left(\bar{F}^N_{j+1,t+1}+\bar{F}^N_{j-1,t+1}\right), \quad && \forall j,t. &
\end{flalign}

The variables $q_{j,t}$ do not appear in $\mathcal{P}^N$, but for any sequence $(p_{j,t})$ we can define $q_{j,t}=\frac{1}{2}(p_{j-1,t-1} + p_{j+1,t-1})-p_{j,t}$ for all $j\in\mathcal{J}''$, $t\geq1$ and similarly for the boundary terms. These problems have complementary slackness conditions
\begin{align}
\label{FCS1} \pjt>0 & \implies  \frac{1}{2} \left( \eta_{j-1,t+1} + \eta_{j+1,t+1} \right)-\eta_{j,t}-\nu_j=\bar{F}^N_{j,t}-\half\left(\bar{F}^N_{j+1,t+1}+\bar{F}^N_{j-1,t+1}\right) \\
\label{FCS2} q_{j,t}>0 &\implies  \eta_{j,t}=0 \\
\label{FCS3} \nu_j>0 & \implies \sum_{t=1}^{\infty} \pjt = U_j. 
\end{align}

The arguments in \cite{Cox:2016ab} show that we have strong duality in the sense that the optimal values of these problems are equal, and both values are obtained by some optimal $p^*, \nu^*, \eta^*$. The original primal-dual pair considered in \cite{Cox:2016ab} optimises over $(p_{j,t})\in l^1(\lambda)=\left\{(x_{j,t}):\medspace \sum_{j,t} |x_{j,t}|\lambda^t <\infty \right\}$ and $(\nu_j, \eta_{j,t})\in \R^{L+1}\times l^{\infty}(\lambda^{-1})$, where $l^{\infty}(\lambda^{-1})=\left\{(y_{j,t}):\medspace \sup_{j,t} |y_{j,t}| \lambda^{-t}<\infty\right\}$ and $\lambda>1$ is a constant. The duality result \cite[Theorem~3.2]{Cox:2016ab} gives dual optimisers $(\nu^*_j, \eta^*_{j,t})\in \R^{L+1}\times l^{\infty}(\lambda^{-1})$, however for the primal optimisers we can only argue that there is an optimal sequence $(p^*_{j,t})\in l^1$, not $l^1(\lambda)$ (\cite[Lemma~3.3]{Cox:2016ab}). 

To ensure that the dual variables are in the true dual space of the primal variables, we require $(\nu^*_j, \eta^*_{j,t})\in \R^{L+1}\times l^{\infty}$. Note that for large $T$ (such that $\bar{F}^N_{j,t}=0$ for all $t\geq T$), $\eta^T_{j,t}=\eta^*_{j,t}\iden\{t<T\}$ gives a feasible sequence $(\eta^T_{j,t})\in l^{\infty}$, and this sequence also gives the same value of the objective function. We can therefore, without loss of generality, restrict our dual problem to $\R^{L+1}\times l^{\infty}$.

With our setup complete, we can now adapt \cite[Theorem~4.2]{Cox:2016ab} to prove a discrete version of \thref{SGTheorem}.

\begin{theorem} \thlabel{Kcaveshape}
The optimal solution of the primal problem $\mathcal{P}^N$, where $\bar{F}^N_{j,t}$ is our discretised LETF function, is given by a sequence $(p^*_{j,t})$ which gives a stopping region for a random walk with the $K$-cave barrier-like property
\begin{align}
\label{leftbarrier} \text{if } q^*_{i,t}>0 \medspace \text{ for some } (i,t) \text{ where }t< K(x_i^N), \text{ then } p^*_{i,s}=0 \medspace \forall s<t, \\
\label{rightbarrier} \text{if } q^*_{i,t}>0 \medspace \text{ for some } (i,t) \text{ where }t> K(x_i^N), \text{ then } p^*_{i,s}=0 \medspace \forall  s>t.
\end{align}
\end{theorem}
\begin{proof}
First consider the inverse-barrier to the left of the curve $K$. To show \eqref{leftbarrier}, suppose we have a feasible solution with $q_{i,t}>0$ and $p_{i,s}>0$ for some $i$ and $s<t<K(x_i^N)$. We take some $0<\eps<\min\{\frac{1}{2}q_{i,t}, p_{i,s}\}$ and show that we can improve our objective function by transferring $\eps$ of the mass that currently leaves $(i,t)$ onto $(i,s)$. We use the $\tilde{p},$ $\tilde{q},$ $\bar{p},$ $\bar{q}$ defined in \cite[Theorem~4.2]{Cox:2016ab}, but repeat them here for convenience. The $\tilde{p},$ $\tilde{q}$ track the $\eps$ of mass leaving $(i,t)$, so
\begin{align*}
\tilde{p}_{i,s}&=\eps, \quad \tilde{q}_{i,s}=-\eps, \\
\tilde{p}_{j,s}&=0, \quad \tilde{q}_{j,s}=0 \quad \forall j\neq i,  \\
\tilde{p}_{j,r+1}&= p_{j,r+1} \frac{\tilde{p}_{j+1,r} + \tilde{p}_{j-1,r}}{p_{j+1,r}+p_{j-1,r}} \qquad \quad \forall j\neq j^N_0,j^N_L, \medspace\forall r\geq s, \\
\tilde{q}_{j,r+1} &= q_{j,r+1} \frac{\tilde{p}_{j+1,r} + \tilde{p}_{j-1,r}}{p_{j+1,r}+p_{j-1,r}} \qquad \quad \forall j\neq j^N_0,j^N_L, \medspace\forall r\geq s,
 \end{align*}
 and similarly for the boundary terms. Using these values, we can write down $\bar{p},$ $\bar{q}$, corresponding to the dynamics of the system after moving this mass:
 \begin{align*}
 \bar{p}_{j,r}&=p_{j,r}, \quad &&\bar{q}_{j,r}= q_{j,r} \quad &&\forall (j,r)\in\{(j,r): \medspace 1\leq r <s\}, \\
 \bar{p}_{j,r}&=p_{j,r}-\tilde{p}_{j,r}, \quad &&\bar{q}_{j,r}=q_{j,r} - \tilde{q}_{j,r} \quad &&\forall (j,r)\in\{(j,r): \medspace s\leq r <t\}, \\
 \bar{p}_{j,r}&=p_{j,r}-\tilde{p}_{j,r}+\tilde{p}_{j,r-(t-s)}, \quad &&\bar{q}_{j,r}=q_{j,r}-\tilde{q}_{j,r}+\tilde{q}_{j,r-(t-s)} \quad &&\forall (j,r)\in\{(j,r): \medspace t\leq r \}.
 \end{align*}
The feasibility of these new probabilities is exactly as in \cite[Lemma~4.3]{Cox:2016ab}. 

Now, $\bar{F}_{j,t}=\left(\bar{h}_{j,t}-k\right)_+$ where $\bar{h}_{Y_t,t}$ is a martingale, so $\sum_{r>s,j} \bar{h}_{j,r}\tilde{q}_{j,r}=\eps \bar{h}_{i,s}$. Let $K_j=K(x_j^N)$, then for any $j$ we have $\{r>s\}=\{s<r\leq K_j-(t-s)\}\cup\{K_j-(t-s)<r\leq K_j\}\cup\{r>K_j\}$. Fix some $j$ such that $s<K_j$, then we have 
\begin{align*}
&\bar{F}_{j,r+t-s} - \bar{F}_{j,r}=\bar{h}_{j,r+t-s} - \bar{h}_{j,r}, \quad &&\text{in } \{s<r\leq K_j-(t-s)\}, \\
&\bar{F}_{j,r+t-s} - \bar{F}_{j,r}=k-\bar{h}_{j,r}\geq \bar{h}_{j,r+t-s} - \bar{h}_{j,r}, \quad &&\text{in } \{K_j-(t-s)<r\leq K_j\}, \\
&\bar{F}_{j,r+t-s} - \bar{F}_{j,r}=0\geq\bar{h}_{j,r+t-s} - \bar{h}_{j,r}, \quad &&\text{in } \{r>K_j\}.
\end{align*}

Combining these, we see that
\begin{align*}
\sum_{j,r} \bar{F}_{j,r} \bar{q}_{j,r} &= \sum_{j,r} \bar{F}_{j,r} q_{j,r} +\eps (\bar{F}_{i,s}-\bar{F}_{i,t}) -\sum_{\mathclap{r>s,j}} \bar{F}_{j,r} \tilde{q}_{j,r}  + \sum_{\mathclap{r>t,j}} \bar{F}_{j,r}\tilde{q}_{j,r-(t-s)} \\
&=\sum_{j,r} \bar{F}_{j,r} q_{j,r} +\eps (\bar{F}_{i,s}-\bar{F}_{i,t}) +\sum_{\mathclap{r>s,j}} \tilde{q}_{j,r} \left( \bar{F}_{j,r+t-s} - \bar{F}_{j,r} \right) \\
&\geq \sum_{j,r} \bar{F}_{j,r} q_{j,r} +\eps (\bar{h}_{i,s}-\bar{h}_{i,t})+\sum_{\mathclap{r>s,j}} \tilde{q}_{j,r} \left( \bar{h}_{j,r+t-s} - \bar{h}_{j,r} \right) \\
&=\sum_{j,r} \bar{F}_{j,r} q_{j,r} + \sum_{r>s,j} \tilde{q}_{j,r}(\bar{h}_{j,r}-\bar{h}_{j,r+t-s})+\sum_{\mathclap{r>s,j}} \tilde{q}_{j,r} \left( \bar{h}_{j,r+t-s} - \bar{h}_{j,r} \right) \\
&= \sum_{j,r} \bar{F}_{j,r} q_{j,r}.
\end{align*}

The right hand barrier \eqref{rightbarrier} is similar, and we use $\hat{p},\hat{q}$ defined in \cite[Theorem~4.2]{Cox:2016ab}. Now we have that $\bar{F}_{j,r}=0$ for $r>K(x_j^N)$ and this simplifies our argument:
\begin{align*}
\sum_{j,r} \bar{F}_{j,r} \hat{q}_{j,r} &= \sum_{j,r} \bar{F}_{j,r} q_{j,r} -\sum_{\mathclap{r>s,j}} \bar{F}_{j,r} \tilde{q}_{j,r}  + \sum_{\mathclap{r>t,j}} \bar{F}_{j,r} \tilde{q}_{j,r+s-t} \\
&=\sum_{j,r} \bar{F}_{j,r} q_{j,r} +\sum_{\mathclap{r>s,j}} (\bar{F}_{j,r-(s-t)}-\bar{F}_{j,r}) \tilde{q}_{j,r} \\
&\geq \sum_{j,r} \bar{F}_{j,r} q_{j,r},
\end{align*}
since $\bar{F}_{j,r}$ is decreasing in $r$.

We have improved the value of our objective function and therefore any solution without this $K$-cave property is suboptimal. Since we know that optimisers exist, they must have this property.
\end{proof}

From \cite[Theorem~3.1, Theorem~4.2]{Cox:2016ab} we know that an optimal solution exists for each $\mathcal{P}^N$ and this is a sequence $(p^{*,N})$ that corresponds to a stopped random walk that is stopped by some almost-deterministic stopping region $\hat{\mathcal{B}}^N$ that takes the form of a $K$-cave barrier. The region $\hat{\mathcal{B}}^N$ is determined by points $\bar{l}^N_j$ and $\bar{r}^N_j$, defined as the largest time $\bar{l}^N_j<K(x^N_j)$ such that $p^{*,N}_{j,t}=0$ $\forall t\leq\bar{l}^N_j$, and similarly the smallest time $\bar{r}^N_j>K(x^N_j)$ such that $p^{*,N}_{j,t}=0$ $\forall t\geq \bar{r}^N_j$. Note that for each $j$ we either have $q^{*,N}_{j,\bar{r}^N_j}>0$, or $q^{*,N}_{j,s}=0$ $\forall s>K(x_j^N)$, and similarly for $\bar{l}^N_j$. These barriers have equivalent stopping regions, $\mathcal{B}^N$, for a Brownian motion, and \cite[Lemma~5.5]{Cox:2016ab} says that these barriers converge to a continuous time $K$-cave barrier $\mathcal{B}^{\infty}$ which embeds $\mu$ into a Brownian motion. From \cite[Lemma~5.6]{Cox:2016ab} we know that the corresponding stopping time is indeed a maximiser of \eqref{OptSEP}, and in fact that the stopped random walks converge to the stopped Brownian motion. In other words, if $\mathrm{P}^N$ is the optimal value of $\mathcal{P}^N$, then $\mathrm{P}^N\rightarrow\sup_{\tau,B_{\tau}\sim\mu}\Ep{F(B_{\tau},\tau)}$, and our discrete barriers converge exactly to an optimal stopping region for \eqref{OptSEP}. This approach therefore reproves \thref{SGTheorem}. Furthermore, we can now look at the convergence of the dual optimisers $\eta^*,\nu^*$.

\subsection{Dual Convergence} \label{dualconvergence}
We know by strong duality that an optimal solution to the linear programming problem is given by the $p,$ $q,$ $\nu,$ $\eta$ that are $\mathcal{P}^N$-feasible and $\mathcal{D}^N$-feasible, and for which the complementary slackness conditions hold. In \thref{Opt} we show that if $\tau$ is such that certain properties of $G,$ $H$ hold, then we have optimality, and as shown in \cite[Section~3.2]{Cox:2016ab}, the complementary slackness conditions here have obvious connections to these properties. Once we have convergence it will guarantee the correct choice of our functions $G,$ $H$ and therefore show that \eqref{Gamma} is both a necessary and sufficient condition for optimality. 

Let $\tau$ be an optimiser of \eqref{OptSEP} of the form of a hitting time of a $K$-cave barrier, which we know exists by \thref{SGTheorem} (or alternatively as a consequence of results in \cite{Cox:2016ab}). Recall that $G(x,t)=-\int_t^{r(x)} M(x,s) \di s -Z(x)$, where $M(x,t)=\Eps{x,t}{\partial^-_t F(B_\tau,\tau)}$, and now we show that our dual optimisers $\eta^{*,N}$ take a similar form. Fix $N$ and let $\mathcal{D}=\left\{(j,t): \medspace p^{*,N}_{j,t}>0\right\}$. For presentation purposes we will drop the dependence on $N$ in much of what follows, so let $\bar{\tau}$ be the stopping law of our random walk $Y$ in the $N$-grid given by the $p^*_{j,t}$ (or $\bar{\tau}^{j,t}$ if $Y$ starts at $(j,t)$). We will also write $\bar{F}^N_{Y_{\bar{\tau}},\bar{\tau}}:=\bar{F}^N(\sqrt{N}Y_{\bar{\tau}},\bar{\tau})$. Then for $(j,t)\in\D$, since we have a positive probability of leaving $(j,t)$, we have $q^*_{Y_{\bar{\tau}},\bar{\tau}}>0$ almost surely, and so by  \eqref{FCS2}, $\eta^{*}_{Y_{\bar{\tau}},\bar{\tau}} =0$. Since we have the interpretation that $\eta^*$ represents $G+H-F$, write $\tilde{\eta}^*=\eta^*+\bar{F}^N$. From \eqref{FCS1} we deduce that
\begin{equation} 
\tilde{\eta}^{*}_{j,t}=\Eps{j,t}{\eta^{*}_{Y_{\bar{\tau}},\bar{\tau}}+\bar{F}^N_{Y_{\bar{\tau}},\bar{\tau}}-\sum_{s=t}^{\bar{\tau}-1} \nu^{*}_{Y_s}} =\Eps{j,t}{\bar{F}^N_{Y_{\bar{\tau}},\bar{\tau}}-\sum_{s=t}^{\bar{\tau}-1} \nu^{*}_{Y_s}}. \label{etatilde}
\end{equation}
Now define a new stopping time as $\left(\bar{\tau}^{-1}\right)^{j,t-1}=\inf \left\{n\geq t-1: \medspace \left(Y_n^{j,t-1},n+1\right)\notin\mathcal{D} \right\}$. By the strong Markov property we see that $\left(\bar{\tau}^{-1}\right)^{j,t-1}=\bar{\tau}^{j,t}-1\geq t-1$, and $Y^{j,t-1}_{\bar{\tau}^{-1}}=Y^{j,t}_{\bar{\tau}}$. Now,
\begin{align*}
\tilde{\eta}^{*}_{j,t-1}&\geq \Eps{j,t-1}{\eta^{*}_{Y_{\bar{\tau}^{-1}},\bar{\tau}^{-1}}+\bar{F}^N_{Y_{\bar{\tau}^{-1}},\bar{\tau}^{-1}}-\sum_{s=t}^{\bar{\tau}^{-1}-1} \nu^{*}_{Y_s}} \quad &\text{by \eqref{Fetasuper}} \\
&\geq \Eps{j,t-1}{\bar{F}^N_{Y_{\bar{\tau}^{-1}},\bar{\tau}^{-1}}-\sum_{s=t}^{\bar{\tau}^{-1}-1} \nu^{*}_{Y_s}} \quad &\text{by \eqref{Fdiscretesuper}} \\
&=\Eps{j,t}{\bar{F}^N_{Y_{\bar{\tau}},\bar{\tau}-1}-\sum_{s=t}^{\bar{\tau}-2} \nu^{*}_{Y_s}}.&
\end{align*}
We then have 
\begin{equation*}
\tilde{\eta}^*_{j,t}-\tilde{\eta}^*_{j,t-1}\leq\Eps{j,t}{\bar{F}^N_{Y_{\bar{\tau}},\bar{\tau}}-\bar{F}^N_{Y_{\bar{\tau}},\bar{\tau}-1} - \nu^*_{Y_{\bar{\tau}-1}}} \le \Eps{j,t}{\bar{F}^N_{Y_{\bar{\tau}},\bar{\tau}}-\bar{F}^N_{Y_{\bar{\tau}},\bar{\tau}-1}}.
\end{equation*}

In a very similar fashion we can find a lower bound, giving us
\begin{equation*}
\Eps{j,t-1}{\bar{F}^N_{Y_{\bar{\tau}},\bar{\tau}+1}-\bar{F}^N_{Y_{\bar{\tau}},\bar{\tau}}}\leq \tilde{\eta}^*_{j,t}-\tilde{\eta}^*_{j,t-1} \leq \Eps{j,t}{\bar{F}^N_{Y_{\bar{\tau}},\bar{\tau}}-\bar{F}^N_{Y_{\bar{\tau}},\bar{\tau}-1}}\leq 0.
\end{equation*}
From the form of $\bar{F}^N$ in the Brownian case (geometric Brownian motion is similar) we deduce that $\bar{F}^N_{j,t}-\bar{F}^N_{j,t-1}=\mathrm{e}^{\beta x^N_j} \left(\cosh\left(\frac{\beta}{\sqrt{N}}\right)\right)^{-t}\left(1-\cosh\left(\frac{\beta}{\sqrt{N}}\right)\right)$. In particular, we have that
\begin{equation*}
N\left(\bar{F}^N_{\lfloor \sqrt{N}x \rfloor,\lfloor Nt \rfloor}-\bar{F}^N_{\lfloor \sqrt{N}x\rfloor,\lfloor Nt \rfloor-1}\right)\rightarrow \partial^-_t F(x,t).
\end{equation*}

In \cite{Cox:2016ab} we show that $\left| \left(\frac{\bar{\tau}^N}{N}, Y_{\bar{\tau}^N}\right) - \left( \tau^N, B_{\tau^N}\right) \right| \xrightarrow{d} 0$, and $\left( \tau^N, B_{\tau^N}\right)\xrightarrow{\P}\left(\tau,B_{\tau}\right)$ as $N\rightarrow\infty$, where $\tau$ is an optimiser of \eqref{OptSEP} and $\tau^N$ is the Brownian hitting time of the $K$-cave barrier $\mathcal{B}^N$. Therefore, since $\bar{F}^N$ and $F$ are bounded in our domain and Lipschitz continuous in time,
\begin{equation*}
N\Eps{\lfloor \sqrt{N}x\rfloor,\lfloor Nt\rfloor}{\bar{F}^N_{Y_{\bar{\tau}},\bar{\tau}}-\bar{F}^N_{Y_{\bar{\tau}},\bar{\tau}-1}}\rightarrow \Eps{x,t}{\partial^-_t F(B_\tau,\tau)}, \quad \text{as } N\rightarrow\infty.
\end{equation*}
We can now find the limit of our dual optimisers $\tilde{\eta}^*$. 

For any $x$, let $\bar{r}^N_x$ denote the left-most point of the right-hand barrier at level $\lfloor \sqrt{N}x\rfloor$ of $\hat{\mathcal{B}}^N$. Then $r(x):=\lim_{N\rightarrow\infty}\frac{\bar{r}^N_x}{N}\in[K(x),\infty]$ is the left-most point of the right hand boundary at $x$ of the limit barrier $\mathcal{B}^{\infty}$.

\begin{lemma}
\thlabel{etaconv}
For any $(x,t)$ in our domain,
\begin{equation*}
\tilde{\eta}^*_{\lfloor \sqrt{N}x\rfloor,\lfloor Nt\rfloor}-\tilde{\eta}^*_{\lfloor \sqrt{N}x\rfloor,\bar{r}^N_x}\rightarrow\int_{r(x)}^t \Eps{x,s}{\partial^-_t F(B_{\tau},\tau)} \di s \quad \text{as } N\rightarrow\infty.
\end{equation*}
\end{lemma}

\begin{proof}
Suppose first $r(x)<\infty$. If $t>r(x)$ then $\exists N_0$ such that $N\geq N_0 \implies Nt>\bar{r}^N_x$ and then $\tilde{\eta}^*_{\lfloor \sqrt{N}x\rfloor,\lfloor Nt\rfloor}=0$ by \eqref{FCS2} and we are done. Suppose $t<r(x)$, then for large $N$ we know by the above that 
\begin{equation*}
 - \sum_{\mathclap{s=\lfloor Nt\rfloor+1}}^{\bar{r}^N_x} \Eps{\lfloor \sqrt{N}x\rfloor,s-1}{\bar{F}^N_{Y_{\ttau},\ttau+1}-\bar{F}^N_{Y_{\ttau},\ttau}} \leq\medspace \tilde{\eta}^*_{\lfloor \sqrt{N}x\rfloor,\lfloor Nt\rfloor}-\tilde{\eta}^*_{\lfloor \sqrt{N}x\rfloor,\bar{r}^N_x}\medspace\leq - \sum_{\mathclap{s=\lfloor Nt\rfloor+1}}^{\bar{r}^N_x} \Eps{\lfloor \sqrt{N}x\rfloor,s}{\bar{F}^N_{Y_{\tilde{\tau}},\tilde{\tau}}-\bar{F}^N_{Y_{\ttau},\ttau-1}}.
\end{equation*}
We look at the convergence of the right-hand side and argue that the other inequality is similar. First note that when $\bar{r}^N_x<\infty$, we know $q_{j,\bar{r}^N_x}>0$, and so by our complementary slackness conditions, $\tilde{\eta}^*_{\lfloor \sqrt{N}x\rfloor,\bar{r}^N_x}=\bar{F}^N_{\lfloor \sqrt{N}x\rfloor,\bar{r}^N_x}=0$, since $\eta^*=0$ in the stopping region. Now,
\begin{align*}
\sum_{\mathclap{s=\lfloor Nt\rfloor+1}}^{\bar{r}^N_x} \Eps{\lfloor \sqrt{N}x\rfloor,s}{\bar{F}^N_{Y_{\tilde{\tau}},\tilde{\tau}}-\bar{F}^N_{Y_{\ttau},\ttau-1}} &= \sum_{s=\frac{\lfloor Nt\rfloor+1}{N}}^{\frac{\bar{r}^N_x}{N}} \Eps{\lfloor \sqrt{N}x\rfloor,\lfloor Ns\rfloor}{\bar{F}^N_{Y_{\tilde{\tau}},\tilde{\tau}}-\bar{F}^N_{Y_{\ttau},\ttau-1}} \\
&=\sum_{s=\frac{\lfloor Nt\rfloor+1}{N}}^{\frac{\bar{r}^N_x}{N}} N\Eps{\lfloor \sqrt{N}x\rfloor,\lfloor Ns\rfloor}{\bar{F}^N_{Y_{\tilde{\tau}},\tilde{\tau}}-\bar{F}^N_{Y_{\ttau},\ttau-1}} \frac{1}{N}\\
&= \int_{\frac{(\lfloor Nt\rfloor+1)}{N}}^{\frac{\bar{r}^N_x}{N}} N\Eps{\lfloor \sqrt{N}x\rfloor,\lfloor Ns\rfloor}{\bar{F}^N_{Y_{\tilde{\tau}},\tilde{\tau}}-\bar{F}^N_{Y_{\ttau},\ttau-1}} \di s \\
&= \left(\int^{r(x)}_t +\int_{\frac{(\lfloor Nt\rfloor+1)}{N}}^t +\int_{r(x)}^{\frac{\bar{r}^N_x}{N}}\right)  N\Eps{\lfloor \sqrt{N}x\rfloor,\lfloor Ns\rfloor}{\bar{F}^N_{Y_{\tilde{\tau}},\tilde{\tau}}-\bar{F}^N_{Y_{\ttau},\ttau-1}} \di s.
\end{align*}
Since we are working in $[x_*,x^*]$ we see that the integrand above is non-positive and bounded below, and also
\begin{align*}
N\left(\bar{F}^N_{j,t}-\bar{F}^N_{j,t-1}\right)&=N\mathrm{e}^{\beta x^N_j} \left(\cosh\left(\frac{\beta}{\sqrt{N}}\right)\right)^{-t}\left(1-\cosh\left(\frac{\beta}{\sqrt{N}}\right)\right)\\
&\geq N\mathrm{e}^{\beta x^*}\left(1-\cosh\left(\frac{\beta}{\sqrt{N}}\right)\right)\\
&\rightarrow -\frac{\beta^2}{2}\mathrm{e}^{\beta x^*},
\end{align*}
as $N\rightarrow\infty$. Then the two remainder integral terms vanish, since
\begin{align*}
\left| \int_{\frac{\lfloor Nt\rfloor+1}{N}}^t N\Eps{\lfloor \sqrt{N}x\rfloor,\lfloor Ns\rfloor}{\bar{F}^N_{Y_{\tilde{\tau}},\tilde{\tau}}-\bar{F}^N_{Y_{\ttau},\ttau-1}} \di s \right| &\leq \left(t-\frac{\lfloor Nt\rfloor+1}{N}\right) \max_s N\Eps{\lfloor \sqrt{N}x\rfloor,\lfloor Ns\rfloor}{\left| \bar{F}^N_{Y_{\tilde{\tau}},\tilde{\tau}}-\bar{F}^N_{Y_{\ttau},\ttau-1}\right|} \\
&\leq \left(t-\frac{\lfloor Nt\rfloor+1}{N}\right)N\mathrm{e}^{\beta x^*}\left(1-\cosh\left(\frac{\beta}{\sqrt{N}}\right)\right) \\
& \rightarrow 0, \quad \text{as } N\rightarrow\infty,
\end{align*}
and similarly for the other integral since $\frac{\bar{r}^N_x}{N}-r(x)\rightarrow 0$. Finally, by the Dominated Convergence Theorem,
\begin{equation*}
-\int^{r(x)}_t N\Eps{\lfloor \sqrt{N}x\rfloor,\lfloor Ns\rfloor}{\bar{F}^N_{Y_{\tilde{\tau}},\tilde{\tau}}-\bar{F}^N_{Y_{\ttau},\ttau-1}} \di s\rightarrow -\int^{r(x)}_t \Eps{x,s}{\partial^-_t F(B_{\tau},\tau)} \di s.
\end{equation*}
The other inequality is similar, and then we conclude by the sandwich theorem.

If $r(x)=\infty$, then the integral on the right hand-side above is still finite since we are working on a bounded domain and $F=0$ for large $t$. In this case the same argument holds once we observe that only finitely many terms in each of our sums can be non-zero.
\end{proof}

We can now prove that our discrete dual optimisers converge to exactly the dual solution we gave earlier, and that we therefore have strong duality in the continuous time problem, but first we look at the effect of the $\tilde{\eta}^*_{\lfloor \sqrt{N}x\rfloor,\bar{r}^N_x}$ term in the above. Recall that we define $\Gamma(x)=\int_{l(x)}^{r(x)} M(x,s) \di s +F(x,l(x))$, so by \thref{etaconv},
\begin{equation*}
\Gamma(x)=\lim_{N\rightarrow\infty} \left(-\tilde{\eta}^*_{\lfloor \sqrt{N}x\rfloor,\bar{l}^N_x}+\tilde{\eta}^*_{\lfloor \sqrt{N}x\rfloor,\bar{r}^N_x}+\bar{F}^N_{\lfloor \sqrt{N} x\rfloor,\bar{l}^N_x}\right)=\lim_{N\rightarrow\infty} \left( \eta^*_{\lfloor \sqrt{N}x\rfloor,\bar{r}^N_x}-\eta^*_{\lfloor \sqrt{N}x\rfloor,\bar{l}^N_x}\right).
\end{equation*}
Then, for $x\in\mathrm{supp}(\mu_r)$, $\eta^*_{\lfloor \sqrt{N} x\rfloor,\bar{l}^N_x}\geq 0$ and $\eta^*_{\lfloor \sqrt{N} x\rfloor,\bar{r}^N_x}=0$ by \eqref{FCS2}, so
\begin{equation*}
\Gamma(x)=\lim_{N\rightarrow\infty} \left( \bar{F}^N_{\lfloor \sqrt{N} x\rfloor,\bar{l}^N_x}-\tilde{\eta}^*_{\lfloor \sqrt{N}x\rfloor,\bar{l}^N_x}\right)=- \lim_{N\rightarrow\infty} \eta^*_{\lfloor \sqrt{N}x\rfloor,\bar{l}^N_x} \leq0 \implies \lim_{N\rightarrow\infty} \tilde{\eta}^*_{\lfloor \sqrt{N}x\rfloor,\bar{r}^N_x}=0=(\Gamma(x))^+.
\end{equation*}
For $x\in\mathrm{supp}(\mu_l)$, $\eta^*_{\lfloor \sqrt{N} x\rfloor,\bar{r}^N_x}\geq0$ and $\eta^*_{\lfloor \sqrt{N} x\rfloor,\bar{l}^N_x}= 0$ by \eqref{FCS2}, so
\begin{equation*}
\Gamma(x)=\lim_{N\rightarrow\infty} \tilde{\eta}^*_{\lfloor \sqrt{N}x\rfloor,\bar{r}^N_x} \geq0 \implies \lim_{N\rightarrow\infty} \tilde{\eta}^*_{\lfloor \sqrt{N}x\rfloor,\bar{r}^N_x}=(\Gamma(x))^+.
\end{equation*}
In particular we have proven the following.

\begin{lemma} \thlabel{etarightconv}
For any $x\in\mathrm{supp}(\mu_l)\cup\mathrm{supp}(\mu_r)$, $\lim_{N\rightarrow\infty} \tilde{\eta}^*_{\lfloor \sqrt{N}x\rfloor,\bar{r}^N_x} = (\Gamma(x))^+$. Furthermore, in the limiting $K$-cave barrier $\mathcal{B}^{\infty}$, \eqref{Gamma} holds.
\end{lemma}

We have shown that the condition \eqref{Gamma} holds in our limiting stopping region, and all that remains to show is that with our functions $G$ and $H$ from \thref{Opt} there is no duality gap.

\begin{theorem}\thlabel{fullconvergence}
With $G(x,t)$, $H(x)$ defined as in \thref{Opt},
\begin{equation*}
\sup_{\tau,B_{\tau}\sim\mu} \Ep{F(B_{\tau},\tau)}=\Ep{G(B_{\tau},\tau)+H(B_{\tau})}.
\end{equation*}
\end{theorem}

\begin{proof}
By \thref{etaconv} and \thref{etarightconv}, $\lim_{N\rightarrow\infty} \tilde{\eta}^*_{\lfloor \sqrt{N}x\rfloor,\lfloor Nt\rfloor}=G^*(x,t)+(\Gamma(x))^+=G(x,t) + H(x)$ for $G$, $H$ as in \thref{Opt} (for $x\notin\mathrm{supp}(\mu_l)\cup\mathrm{supp}(\mu_r)$ we can ensure this by our freedom of choice of $H(x)$). We can write $\tilde{\eta}^*_{j,t}=\tilde{\eta}^{*,G}_{j,t}+\tilde{\eta}^{*,H}_{j,t}$ such that $\tilde{\eta}^{*,G}_{j,t}$ is a martingale and $\lim_{N\rightarrow\infty} \tilde{\eta}^{*,G}_{\lfloor \sqrt{N}x\rfloor,\lfloor Nt\rfloor}=G(x,t)$ for any $(x,t)$. Then clearly $\lim_{N\rightarrow\infty} \tilde{\eta}^{*,H}_{\lfloor \sqrt{N}x\rfloor,\lfloor Nt\rfloor}=H(x)$ for any $(x,t)$, and so $\tilde{\eta}^{*,H}_{j,t}=\tilde{\eta}^{*,H}_j$ is independent of $t$. 

When $N$ is sufficiently large, for every $j\neq j^N_0,j^N_L$ there is some $t$ such that $\pjt>0$, so from \eqref{FCS1} we have
\begin{equation*}
\nu_j= \half\left(\tilde{\eta}^*_{j+1,t+1}+\tilde{\eta}^*_{j-1,t+1}\right)-\tilde{\eta}^*_{j,t} =\half\left(\tilde{\eta}^{*,H}_{j+1}+\tilde{\eta}^{*,H}_{j+1}\right)-\tilde{\eta}^{*,H}_{j}.
\end{equation*}

From the ideas in \cite[Section~3.2]{Cox:2016ab} we suspect that $N\nu_{\lfloor \sqrt{N}x\rfloor} \rightarrow\half H''(x)$ as $N\rightarrow\infty$. Since we cannot argue the convergence of derivatives, the corresponding summation is
\begin{align*}
\sum_{m=1}^{i} \sum_{k=1}^{m} \nu_{j_k} &= \sum_{m=1}^{i} \sum_{k=1}^{m} \half\left(\tilde{\eta}^{*,H}_{j_{k+1}}+\tilde{\eta}^{*,H}_{j_{k-1}}\right)-\tilde{\eta}^{*,H}_{j_k} \\
&= \sum_{m=1}^{i} \left(\half \sum_{k=2}^{m+1} \tilde{\eta}^{*,H}_{j_k} +\half \sum_{k=0}^{m-1} \tilde{\eta}^{*,H}_{j_k} -\half \sum_{k=1}^{m} \tilde{\eta}^{*,H}_{j_k}\right) \\
&= \half \sum_{m=1}^i \left( \left(\tilde{\eta}^{*,H}_{j_{m+1}}-\tilde{\eta}^{*,H}_{j_m}\right)-\left(\tilde{\eta}^{*,H}_{j_1}-\tilde{\eta}^{*,H}_{j_0}\right)\right) \\
&=\half\left(\tilde{\eta}^{*,H}_{j_{i+1}}-\tilde{\eta}^{*,H}_{j_1}\right) -\half i\left(\tilde{\eta}^{*,H}_{j_1}-\tilde{\eta}^{*,H}_{j_0}\right),
\end{align*}
and in particular,
\begin{equation*}
\lim_{N\rightarrow\infty} \sum_{m=1}^{\lfloor \sqrt{N}x\rfloor} \sum_{j=j_0}^{m} \nu_{j} = \half\left(H(x)-H(x_*)\right) -\half \lim_{N\rightarrow\infty} \lfloor \sqrt{N}x\rfloor \left(\tilde{\eta}^{*,H}_{j_1}-\tilde{\eta}^{*,H}_{j_0}\right).
\end{equation*}

Our aim is to rewrite $\sum_{j\in\mathcal{J}'} \nu_j U_j$ to incorporate this double sum by an integration by parts type argument and work instead with the derivatives of $U$. Let $V_{j_i}=U_{j_{i+1}}-U_{j_i}$, $W_{j_i}=V_{j_{i+1}}-V_{j_i}$ for $i=0,\ldots,L-1$, $V_{j_L}=W_{j_L}=0$, and $\nu_0=0$. Then, noting that $U_{j_L}=0$,
\begin{align*}
\sum_{j\in\mathcal{J}'} \nu_j U_j&=\sum_{i=1}^{L-1} \left(\sum_{k=0}^i \nu_{j_k} - \sum_{k=0}^{i-1} \nu_{j_k} \right) U_{j_i} \\
&=- \sum_{i=1}^{L-1} \left(\sum_{k=0}^i \nu_{j_k}\right)V_{j_i} + \left(\sum_{k=1}^{L-1} \nu_{j_k}\right)U_{j_L} \\
&=- \sum_{i=1}^{L-1} \left(\sum_{m=0}^i \sum_{k=0}^m \nu_{j_k} - \sum_{m=0}^{i-1} \sum_{k=0}^m \nu_{j_k}\right) V_{j_i} \\
&=\sum_{i=1}^{L-2} \left( \sum_{m=1}^i \sum_{k=1}^m \nu_{j_k}\right) W_{j_i} -\left(\sum_{i=1}^{L-1} \sum_{k=1}^i \nu_{j_k}\right) V_{j_{L-1}}.
\end{align*}

Substituting in our expression for the double summation of $\nu$, the first term becomes
\begin{align*}
\sum_{i=1}^{L-2} \left( \sum_{m=1}^i \sum_{k=1}^m \nu_{j_k}\right) W_{j_i}& = \sum_{i=1}^{L-2}\left(\half\left(\tilde{\eta}^{*,H}_{j_{i+1}}-\tilde{\eta}^{*,H}_{j_1}\right) -\half i\left(\tilde{\eta}^{*,H}_{j_1}-\tilde{\eta}^{*,H}_{j_0}\right)\right)W_{j_i}\\
&= \sum_{i=1}^{L-2} \half\left(\tilde{\eta}^{*,H}_{j_{i+1}}-\tilde{\eta}^{*,H}_{j_1}\right)W_{j_i} -\sum_{i=1}^{L-2} \half i\left(\tilde{\eta}^{*,H}_{j_1}-\tilde{\eta}^{*,H}_{j_0}\right)\left(V_{j_{i+1}}-V_{j_i}\right) \\
&= \sum_{i=1}^{L-2} \half\left(\tilde{\eta}^{*,H}_{j_{i+1}}-\tilde{\eta}^{*,H}_{j_1}\right)W_{j_i} - \half\left(\tilde{\eta}^{*,H}_{j_1}-\tilde{\eta}^{*,H}_{j_0}\right) \left((L-1)V_{j_{L-1}}-\sum_{i=1}^{L-1} V_{j_i}\right) \\
&= \sum_{i=1}^{L-2} \half\left(\tilde{\eta}^{*,H}_{j_{i+1}}-\tilde{\eta}^{*,H}_{j_1}\right)W_{j_i} - \half\left(\tilde{\eta}^{*,H}_{j_1}-\tilde{\eta}^{*,H}_{j_0}\right) \left(V_{j_0}+(L-1)V_{j_{L-1}}\right).
\end{align*}
For the second term,
\begin{align*}
\left(\sum_{i=1}^{L-1} \sum_{k=1}^i \nu_{j_k}\right) V_{j_{L-1}}&= \half\left(\tilde{\eta}^{*,H}_{j_{L}}-\tilde{\eta}^{*,H}_{j_1}\right) V_{j_{L-1}}-\half \left(\tilde{\eta}^{*,H}_{j_1}-\tilde{\eta}^{*,H}_{j_0}\right) (L-1)V_{j_{L-1}},
\end{align*}
and so
\begin{align*}
\sum_{j\in\mathcal{J}'} \nu_j U_j&= \sum_{i=1}^{L-2} \half\left(\tilde{\eta}^{*,H}_{j_{i+1}}-\tilde{\eta}^{*,H}_{j_1}\right)W_{j_i}-\half\left(\tilde{\eta}^{*,H}_{j_1}-\tilde{\eta}^{*,H}_{j_0}\right) V_{j_0} -\half\left(\tilde{\eta}^{*,H}_{j_{L}}-\tilde{\eta}^{*,H}_{j_1}\right) V_{j_{L-1}}.
\end{align*}

To work with the derivatives of the potential, we now approximate it by smooth functions. From our choice of $U^N$, for each $x$ we know $\frac{1}{\sqrt{N}}U^N(\lfloor \sqrt{N}x\rfloor)\rightarrow U_{\delta_0}(x) - U_{\mu}(x)$, the difference in the potential functions of the distributions $\delta_0$ and $\mu$. These function are continuous and concave so by the Stone-Weierstrass theorem there exists a decreasing sequence of functions, $(\tilde{U}^n)_n$, in $C^{\infty}$ converging uniformly to $U_{\delta_0} - U_{\mu}$ with $\tilde{U}^n(x_*)=\tilde{U}^n(x^*)=0$ for all $n$. For a given $n$ we can find discrete approximations $\tilde{U}^{n,N}$ of $\tilde{U}^n$ (and the associated $\tilde{V}^{n,N}$, $\tilde{W}^{n,N}$) such that $\frac{1}{\sqrt{N}}\tilde{U}^{n,N}_{\lfloor \sqrt{N}x\rfloor}\rightarrow\tilde{U}^n(x)$, $\tilde{V}^{n,N}_{\lfloor \sqrt{N}x\rfloor}\rightarrow\frac{\di \tilde{U}^n}{\di x}(x)$, and $\sqrt{N}\tilde{W}^{n,N}_{\lfloor \sqrt{N}x\rfloor}\rightarrow\frac{\di^2 \tilde{U}^n}{\di x^2}(x)$ for all $x$ as $N\rightarrow\infty$. We can also, without loss of generality, choose $\tilde{U}^{n,N}$ such that $\tilde{U}^{n,N}_{j_0}=\tilde{U}^{n,N}_{j_L}=0$ and $\tilde{U}^{n,N}\geq \tilde{U}^n$.

Then, by the above,
\begin{align*}
\sum_{j\in\mathcal{J}'} \nu_j \tilde{U}^{n,N}_j&=\sum_{i=1}^{L-2} \half\left(\tilde{\eta}^{*,H}_{j_{i+1}}-\tilde{\eta}^{*,H}_{j_1}\right)\tilde{W}^{n,N}_{j_i}  -\half\left(\tilde{\eta}^{*,H}_{j_1}-\tilde{\eta}^{*,H}_{j_0}\right) \tilde{V}^{n,N}_{j_0} - \half\left(\tilde{\eta}^{*,H}_{j_{L}}-\tilde{\eta}^{*,H}_{j_1}\right) \tilde{V}^{n,N}_{j_{L-1}} \\
&=\int_{\frac{1}{\sqrt{N}}}^{\frac{L-2}{\sqrt{N}}} \half\left(\tilde{\eta}^{*,H}_{\lfloor \sqrt{N}x\rfloor+1}-\tilde{\eta}^{*,H}_{j_1}\right)\sqrt{N}\tilde{W}^{n,N}_{\lfloor \sqrt{N}x\rfloor} \di x-\half\left(\tilde{\eta}^{*,H}_{j_1}-\tilde{\eta}^{*,H}_{j_0}\right) \tilde{V}^{n,N}_{j_0}\\
&\qquad\qquad\qquad\qquad\qquad\qquad\qquad\qquad\qquad\qquad\qquad\qquad\qquad- \half\left(\tilde{\eta}^{*,H}_{j_{L}}-\tilde{\eta}^{*,H}_{j_1}\right) \tilde{V}^{n,N}_{j_{L-1}} \\
&\rightarrow \int \half \left(H(x)-H(x_*)\right)\frac{\di^2 \tilde{U}^n}{\di x^2}(x)\dx -\half\left(H(x^*)-H(x_*)\right)\frac{\di \tilde{U}^n}{\di x}(x^*), \quad \text{as }N\rightarrow\infty.
\end{align*}
Since $H$ is convex, it is differentiable almost everywhere and has a second derivative in the sense of distributions. Using integration by parts again,
\begin{align*}
\int \half \left(H(x)-H(x_*)\right)\frac{\di^2 \tilde{U}^n}{\di x^2}(x)\dx&= \int \half H(x)\frac{\di^2 \tilde{U}^n}{\di x^2}(x) \dx - \half H(x_*)\left(\frac{\di \tilde{U}^n}{\di x}(x^*)-\frac{\di \tilde{U}^n}{\di x}(x_*)\right) \\
&=\int \half H''(x)\tilde{U}^n(x) \dx +\half\left(H(x^*)-H(x_*)\right)\frac{\di \tilde{U}^n}{\di x}(x^*).
\end{align*}
Therefore,
\begin{equation*}
\sum_{j\in\mathcal{J}'} \nu_j \tilde{U}^{n,N}_j\rightarrow \int \half H''(x)\tilde{U}^n(x) \dx \quad \text{as }N\rightarrow\infty,
\end{equation*}
and so by monotone convergence
\begin{equation*}
\lim_{n,N\rightarrow\infty}\sum_{j\in\mathcal{J}'} \nu_j \tilde{U}^{n,N}_j= \int \half H''(x) U(x) \dx 
=\int \half H''(x) \left(U_{\delta_0}(x)-U_{\mu}(x)\right) \dx
=\Ep{H(B_{\tau})}-H(B_0).
\end{equation*}

By our choice of approximation we know that $\lim_{N\rightarrow\infty} U^N= U\leq \tilde{U}^n =\lim_{N\rightarrow\infty} \tilde{U}^{n,N}$ for all $n$, and so without loss of generality we can choose $\tilde{U}^{n,N}\geq U^N$ for large $n,N$. Then, since $\nu_j\geq 0$ for all $j$, by monotone convergence it follows that
\begin{equation*}
\left| \sum_{j\in\mathcal{J}'} \nu_j \tilde{U}^{n,N}_j -\sum_{j\in\mathcal{J}'} \nu_j U^n \right| \rightarrow 0, \quad \text{as }n,N\rightarrow\infty.
\end{equation*}
Finally recall that $\mathrm{D}^N=\sum_{j\in\mathcal{J}'} \nu_j U_j + \half\left(\eta_{j^*+1,1}+\eta_{j^*-1,1}\right) +\half\left(\bar{F}^N_{j^*+1,1}+\bar{F}^N_{j^*-1,1}\right)$, so
\begin{equation*}
\lim_{N\rightarrow\infty}\mathrm{D}^N= \Ep{H(B_{\tau})}-H(B_0) + G^*(B_0,0) = \Ep{H(B_{\tau})} + G(B_0,0) = \Ep{G(B_{\tau},\tau)+H(B_{\tau})}.
\end{equation*}
Then by the above and the results of \cite{Cox:2016ab},
\begin{equation*}
\sup_{\tau,B_{\tau}\sim\mu} \Ep{F(B_{\tau},\tau)}=\lim_{N\rightarrow\infty} \mathrm{P}^N=\lim_{N\rightarrow\infty} \mathrm{D}^N=G(B_0,0)+\Ep{H(B_{\tau})}.
\end{equation*}
\end{proof}

\begin{theorem}
For a $K$-cave stopping time $\tau$ given by curves $l$, $r$, the condition \eqref{Gamma} is necessary for optimality.
\end{theorem}

\begin{proof}
By \thref{fullconvergence} our functions $G$ and $H$ give no duality gap. We know that $G(x,t)+H(x)\geq F(x,t)$ everywhere, but also
\begin{align*}
\Gamma(x)>0 &\implies G(x,r(x))+H(x)>F(x,r(x)), \\
\Gamma(x)<0 &\implies G(x,l(x))+H(x)>F(x,l(x)),
\end{align*}
so if \eqref{Gamma} does not hold then $\Ep{F(B_{\tau},\tau)}>G(B_0,0)+\Ep{H(B_{\tau})}$, contradicting \thref{fullconvergence}.
\end{proof}

\subsection{An Additional Property of the Barrier}
We have seen that the linear programming approach to this problem allows us to recover the condition \eqref{Gamma}, but it also reveals additional information about our continuous problem. As mentioned previously, for any dual optimisers $(\nu^*_j,\eta^*_{j,t})\in\R^{L+1}\times l^{\infty}$, the sequence $(\nu^*_j, \eta^T_{j,t})\in\R^{L+1}\times l^{\infty}$, where $\eta^T_{j,t}=\eta^*_{j,t} \iden\{t\leq T\}$, is also dual feasible when $T$ is such that $\bar{F}^N_{j,t}=0$ for $t\geq T$. Furthermore, this new dual solution is also optimal. Since we work on the bounded domain $[x_*,x^*]$, there exists $T^*=\min\{t:\medspace \bar{F}^N_{j,t}=0 \medspace \forall t\geq T, \forall j\}$. Anything that happens after $T^*$ does not affect our payoff, and we therefore have some freedom past this point. We can also see this from our proof of \thref{SGTheorem} if we work on $[x_*,x^*]$. For $t\geq T^*$ we have $L^K_{\infty}(B)=0$, so we have equality in the primary optimisation problem \eqref{SGleq} and require the secondary problem \eqref{SG2} to get the $K$-cave barrier shape. 

In the discrete problem this freedom arises in the following way: if $p_{j,t}>0$ for some $j$ and $t\geq T^*$ then we can stop mass at $(j,t)$, decreasing our local time everywhere (so remaining primal-feasible) without affecting optimality. This allows us to prove the following.

\begin{lemma}
Let $\mu^N_l$ and $\mu^N_r$ be the distributions embedded by our optimal $(p_{j,t})$ to the left and right of $K$ respectively. Then for any $j$,
\begin{equation*}
j\in\mathrm{supp}(\mu^N_l) \medspace \implies \medspace \bar{r}^N_j\leq T^*.
\end{equation*}
In particular, this means that $\mathrm{supp}(\mu^N_r)=\mathrm{supp}(\mu^N)$.
\end{lemma}

\begin{proof}
Take $j\in\mathrm{supp}(\mu^N_l)$, so $q_{j,\bar{l}^N_j-1}>0$, $p_{j,\bar{l}^N_j-1}=0$ and suppose that $T^*<\bar{r}^N_j<\infty$, so $p_{j,\bar{r}^N_j}=0$ and $p_{j,\bar{r}^N_j-1}>0$. Let $\eps=\min\{q_{j,\bar{l}^N_j-1},p_{j,\bar{r}^N_j-1}\}$. We define new primal variables corresponding to stopping $\eps$ of mass at $(j,\bar{r}^N_j-1)$, and releasing $\eps$ of mass at $(j,\bar{l}^N_j-1)$ which we stop after one step. Let
\begin{align*}
\bar{p}_{j,\bar{l}^N_j-1}&=\eps, \quad &\bar{q}_{j,\bar{l}^N_j-1}&=q_{j,\bar{l}^N_j-1}-\eps,\\
\bar{p}_{j+1,\bar{l}^N_j}&=p_{j+1,\bar{l}^N_j}, \quad &\quad\bar{q}_{j+1,\bar{l}^N_j}&=\bar{q}_{j+1,\bar{l}^N_j}+\frac{\eps}{2}, \\
\bar{p}_{j-1,\bar{l}^N_j}&=p_{j-1,\bar{l}^N_j}, \quad &\quad\bar{q}_{j-1,\bar{l}^N_j}&=\bar{q}_{j-1,\bar{l}^N_j}+\frac{\eps}{2}, \\
\bar{p}_{j,\bar{r}^N_j-1}&=p_{j,\bar{r}^N_j-1}-\eps, \quad &\bar{q}_{j,\bar{r}^N_j-1}&=q_{j,\bar{r}^N_j-1}+\eps,\\
\bar{p}_{j,r}&=p_{j,r}-\tilde{p}_{j,r} \quad \forall r>s, \quad &\bar{q}_{j,r}&=q_{j,r}+\tilde{q}_{j,r} \quad \forall r>s,\\
\bar{p}_{j,r}&=p_{j,r} \quad \text{otherwise}, \quad &\bar{q}_{j,r}&=q_{j,r} \quad \text{otherwise},
\end{align*}
where the $\tilde{p}_{j,r}$ are defined as in \thref{Kcaveshape} with $s=\bar{r}^N_j-1$. We can check that these new variables are primal feasible, noting in particular that $\sum_t \bar{p}_{j,t} \leq \sum_t p_{j,t}$. Furthermore, $\bar{F}^N$ is a submartingale, and so by releasing mass at $(j,\bar{l}^N_j-1)$ we improve our payoff: $\sum \bar{F}^N_{j,t}\bar{q}_{j,t}\geq\sum \bar{F}^N_{j,t}q_{j,t}$.

We can repeat this process until either $q_{j,\bar{l}^N_j-1}=0$ or $p_{j,\bar{r}^N_j-1}=0$. If $p_{j,\bar{r}^N_j-1}=0$ first then we have moved $\bar{r}^N_j\rightarrow \bar{r}^N_j-2$ and can repeat the above from this new value of $\bar{r}^N_j$. Similarly, if $q_{j,\bar{l}^N_j-1}=0$ first then we have moved $\bar{l}^N_j\rightarrow \bar{l}^N_j+2$ and can repeat the above. This will continue until either $\bar{r}^N_j\leq T^*$ or $j\notin\mathrm{supp}(\mu^N_l)$. Since we improve our payoff at each step, any optimiser must have one of these properties at each $j$.

If $\bar{r}^N_j=\infty$ for some $j\in\mathrm{supp}(\mu^N_l)$ then we can run the above argument with any $t>T^*$ in place of $\bar{r}^N_j$ and come to the same conclusion. In particular we have a right-hand barrier whenever $j\in\mathrm{supp}(\mu^N_l)$, and obviously also for $j\in\mathrm{supp}(\mu^N)\setminus\mathrm{supp}(\mu^N_l)$.
\end{proof}

The conclusion $\mathrm{supp}(\mu^N_r)=\mathrm{supp}(\mu^N)$ means that we have a right-hand barrier at any atom of $\mu^N$, so in particular $\tilde{\eta}^*_{j,\bar{r}^N_j}=0$ for all $j\in\mathrm{supp}(\mu^N)$. Then, for any $x\in\mathrm{supp}(\mu)$,
\begin{equation*}
\Gamma(x) = \lim_{N\rightarrow\infty} \left(-\tilde{\eta}^*_{\lfloor \sqrt{N}x\rfloor,\bar{l}^N_x}+\tilde{\eta}^*_{\lfloor \sqrt{N}x\rfloor,\bar{r}^N_x}+\bar{F}^N_{\lfloor \sqrt{N} x\rfloor,\bar{l}^N_x}\right)=\lim_{N\rightarrow\infty} -\eta^*_{\lfloor \sqrt{N}x\rfloor,\bar{l}^N_x} \leq 0.
\end{equation*}

In particular, \eqref{Gamma} now becomes
\begin{equation}
\label{Gamma'} \tag{$\Gamma'$}
\begin{aligned}
\Gamma(x)\leq 0 \quad &\mu \text{-a.s.} \\
\Gamma(x)=0 \quad &\mu_l \text{-a.s.}
\end{aligned}
\end{equation}
and we have proved the following.

\begin{theorem}
For a $K$-cave stopping time $\tau$ given by curves $l$, and $r$, the condition \eqref{Gamma'} is both necessary and sufficient for optimality.
\end{theorem}

\section{Uniqueness}
In \thref{SGTheorem} we proved that there is a $K$-cave barrier whose stopping time solves \eqref{OptSEP}, however we also argued in Section \ref{nonunique} that there are multiple $K$-cave barriers solving \eqref{SEP}. We have now characterised the optimal solutions and can ask if there are multiple $K$-cave barriers that are also optimal.

Similarly to Loynes \cite{Loynes:1970aa} we define regular $K$-cave barriers. Take a $K$-cave barrier $\mathcal{R}$ with boundary curves $l$ and $r$. Recall that we define $x^*$ to be the smallest $x$ such that $\mu((x,\infty))=0$, and $x_*$ the largest $x$ such that $\mu((-\infty,x_*))=0$.
\begin{definition}
The $K$-cave barrier $\mathcal{R}$ is \emph{regular} if 
\begin{itemize}
\item $l$ is increasing,
\item $l(x)=K(x)=r(x)$ for all $x\notin[x_*,x^*]$ (where $l$ and $K$ exist),
\item $r(x)=0$ for all $x<x_*\wedge\frac{1}{\beta}\ln k$.
\end{itemize}
\end{definition}

\begin{theorem}
There is a unique regular $K$-cave barrier whose stopping time solves \eqref{OptSEP}.
\end{theorem}

\begin{proof}
Suppose $\tau$, $\sigma$ are both optimisers of \eqref{OptSEP} and hitting times of $K$-cave barriers with continuation regions $\mathcal{D}^{\tau}$ and $\mathcal{D}^{\sigma}$ respectively. By \thref{fullconvergence} these stopping times have dual optimisers $G^{\tau}$, $H^{\tau}$ and $G^{\sigma}$, $H^{\sigma}$, where the functions $G^{\tau}$, $G^{\sigma}$ take the form $G(x,t)=-\int_t^{r(x)} M(x,s) \di s$ for the corresponding $r$ and $M$. Then, 
\begin{align*}
\Ep{F(B_{\tau},\tau)}&=\Ep{G^{\tau}(B_{\tau},\tau)+H^{\tau}(B_{\tau})} \\
&=G^{\tau}(B_0,0)+\int H^{\tau}(x) \mu(\dx) \\
&\geq \Ep{G^{\tau}(B_{\sigma}, \sigma)+H^{\tau}(B_{\sigma})} \\
&\geq \Ep{F(B_{\sigma},\sigma)},
\end{align*}
since $G^{\tau}(B_t,t)$ is a supermartingale, and $G^{\tau}(x,t)+H^{\tau}(x)\geq F(x,t)$ everywhere. However, since both stopping times are optimisers, $\Ep{F(B_{\tau},\tau)}=\Ep{F(B_{\sigma},\sigma)}$, and we have equality in the above, so
\begin{equation*}
\Ep{F(B_{\sigma},\sigma)}=\Ep{G^{\tau}(B_{\sigma}, \sigma)+H^{\tau}(B_{\sigma})}.
\end{equation*}

In Section \ref{optimality} we argue that $G^{\tau}(x,t)+H^{\tau}(x)\geq F(x,t)$ since for $M^{\tau}(x,s)=\Eps{(x,s)}{\partial^-_t F\left(B_{\tau},\tau\right)}$ we have
\begin{align*}
t<K(x) \quad & \implies \quad F_t(x,t)=-\frac{\beta^2}{2}h(x,t)\leq M^{\tau}(x,t)\\
t>K(x) \quad & \implies \quad F_t(x,t)=0\geq M^{\tau}(x,t).
\end{align*}
It is easy to see that these inequalities are strict on $\mathcal{S}^{\tau}_{\mathcal{D}}=\left\{(x,t)\in\mathcal{D}^{\tau}:\medspace \P^{(x,t)}\left(F(B_{\tau},\tau)>0\right)>0\right\}$, and so $G^{\tau}(B_{\sigma}, \sigma)+H^{\tau}(B_{\sigma})>F(B_{\sigma},\sigma)$ for $(B_{\sigma},\sigma)\in\mathcal{S}^{\tau}_{\mathcal{D}}$. Therefore $(B_{\sigma},\sigma)\notin\mathcal{S}^{\tau}_{\mathcal{D}}$ almost surely, and similarly $(B_{\tau},\tau)\notin\mathcal{S}^{\sigma}_{\mathcal{D}}$ almost surely. In particular, this means that the inverse barriers given by $l^{\tau}$ and $l^{\sigma}$ must coincide. 

Furthermore, the argument of Loynes \cite{Loynes:1970aa} proves that for a given inverse barrier bounded by $l^{\tau}$, there is a unique barrier given by some $r^{\tau}$ that gives the correct embedding. This argument runs as follows: suppose we have two Root barriers $\mathcal{R}_1$ and $\mathcal{R}_2$ given by curves $r_1$ and $r_2$ respectively. If our inverse barrier is $\mathcal{R}$, then $B_{\tau_{\mathcal{R}}\wedge\tau_{\mathcal{R}_1}}\sim\mu$ and $B_{\tau_{\mathcal{R}}\wedge\tau_{\mathcal{R}_2}}\sim\mu$. We can consider $\mathcal{R}_0=\mathcal{R}_1\cup\mathcal{R}_2$, or $r_1\wedge r_2$, and show that $B_{\tau_{\mathcal{R}}\wedge\tau_{\mathcal{R}_0}}\sim\mu$ also. By taking the union of the two barriers we increase the area of the stopping region and therefore ensure that no extra paths can be embedded at $\mathcal{R}$. Also, if we have points $\underaccent{\bar}{x}$, $\bar{x}$ such that $r_1(x)\leq r_2(x)$ on $(\underaccent{\bar}{x},\bar{x})$, then less mass is embedded in $(\underaccent{\bar}{x},\bar{x})$ by $\mathcal{R}_0$ than $\mathcal{R}_1$, so overall less mass is embedded in $(\underaccent{\bar}{x},\bar{x})$ by $\tau_{\mathcal{R}}\wedge\tau_{\mathcal{R}_0}$ than $\tau_{\mathcal{R}}\wedge\tau_{\mathcal{R}_1}$. Similarly at points where $r_2(x)\leq r_1(x)$ we have that  less mass is embedded by $\tau_{\mathcal{R}}\wedge\tau_{\mathcal{R}_0}$ than $\tau_{\mathcal{R}}\wedge\tau_{\mathcal{R}_2}$. Then on any interval $A$, $\P\left( B_{\tau_{\mathcal{R}}\wedge\tau_{\mathcal{R}_0}}\in A\right)\leq \mu(A)$, but since both of these distributions are probability measures we must in fact have equality.

This shows that $\mathcal{R}_0$ also embeds the correct distribution, so $B_{\tau_{\mathcal{R}}\wedge\tau_{\mathcal{R}_0}}\sim\mu$, and therefore $\Ep{\tau_{\mathcal{R}}\wedge\tau_{\mathcal{R}_1}\wedge\tau_{\mathcal{R}_2}}=\Ep{\tau_{\mathcal{R}}\wedge\tau_{\mathcal{R}_0}}=\Ep{B^2_{\tau_{\mathcal{R}}\wedge\tau_{\mathcal{R}_0}}}=\int x^2 \mu(\dx)=\Ep{\tau_{\mathcal{R}}\wedge\tau_{\mathcal{R}_1}}$, so $\tau_{\mathcal{R}}\wedge\tau_{\mathcal{R}_1}\wedge\tau_{\mathcal{R}_2}=\tau_{\mathcal{R}}\wedge\tau_{\mathcal{R}_1}$ almost surely. We can then conclude that $\mathcal{R}_1$ and $\mathcal{R}_2$ are equivalent as in \cite{Loynes:1970aa}.
\end{proof}

We now summarise what we know of the uniqueness of these barriers:
\begin{itemize}
\item There may be many (regular) $K$-cave barriers whose hitting times solve \eqref{SEP}.

\item There is exactly one regular $K$-cave barrier whose hitting time solves \eqref{OptSEP}, and this is the regular $K$-cave barrier satisfying \eqref{Gamma'}.

\item All other solutions of \eqref{OptSEP} have the same stopping region as the regular $K$-cave barrier solution, $\tau$, on $\mathcal{S}^{\tau}=\left\{(x,t):\medspace \P^{(x,t)}\left(F(B_{\tau},\tau)>0\right)>0\right\}$. In particular they have the same inverse barrier.
\end{itemize}

In the spirit of \cite{Loynes:1970aa} we could say that two stopping regions are $\tau$-equivalent if they agree on $\mathcal{S}^{\tau}$, and then any region whose hitting time solves \eqref{OptSEP} is $\tau$-equivalent to the $K$-cave optimiser $\tau$.

\section{Conclusions}
In this paper we have given a new solution to the Skorokhod embedding problem that arises when considering model-independent bounds on the price of European call options on a leveraged exchange traded fund. Unlike many classical solutions to \eqref{SEP}, stopping times of this form are not unique and we have used two very different methods to find the form of the superhedging portfolio in order to identify the optimal stopping region. One method involves a PDE approach to suggest the form of such a portfolio and then purely probabilistic arguments to confirm the sufficient condition for optimality with this portfolio. The other approach originates from the idea of considering a discretised version of \eqref{SEP} as a primal-dual linear programming pair in \cite{Cox:2016ab} and then constructing a superhedging strategy as the limit of the dual optimisers.

The techniques used are not specific to our choice of payoff, for example all results in Section \ref{discrete} hold for payoff of the form $F(x,t)=\left(h(x,t)-k\right)_+$ where $h(x,t)$ is decreasing in time and such that $h(X_t,t)$ is a martingale. These results are also true when we consider the case of the cave embedding, and furthermore, the probabilistic approach of \thref{Opt} also holds with the cave payoff. The cave payoff is of a very different form to the European call option payoff we work with in this paper, and this raises some natural questions. Firstly, are there other solutions to \eqref{SEP} which can be seen as the combination of Root and Rost barriers? Secondly, if there are other solutions, is \eqref{Gamma} the correct condition to choose the unique pair satisfying some maximisation problem? Given a maximisation problem we can characterise its optimiser geometrically using the monotonicity principle of \cite{Beiglboeck:2013aa}, and if the optimal stopping time does take the form of a hitting time for Brownian motion then the ideas used in Section \ref{optimality} and Section \ref{discrete} here should be applicable. 

There are other natural questions that arise from this problem, and many of them have been asked, and answered, for other embeddings:
\begin{itemize}
\item Is it possible to generalise our results to general starting distributions? In this paper we always consider a (geometric) Brownian motion started at some fixed point, but it should be possible to consider general starting distributions, indeed the results used from \cite{Cox:2016ab} appear to hold for more general starting distributions. When considering just the Rost barrier, for a true hitting time solution we require the supports of the initial and target distributions to be disjoint (see \cite{Cox:2013ac}). The inclusion of the extra Root barrier will allow us to embed at points in the support of the initial distribution, but there may be technicalities. 
 
\item In \cite{Cox:2015ab}, the authors consider the Root solution to the multi-marginal Skorokhod embedding problem through an optimal stopping approach, and Rost barriers have also been considered in terms of optimal stopping problems in \cite{DeAngelis:2015aa}. Is there a similar optimal stopping formulation for this problem? Can an optimal stopping setup, or any other approach, give a multi-marginal result for cave or $K$-cave barriers?

\item Can we use similar methods to find a robust lower bound on the price of our option? The monotonicity principle calculation in this case will be the exact opposite, meaning that the optimal stopping region will look like the continuation region of a $K$-cave barrier. For target distributions with full support we would then require some external randomisation in the stopping region, perhaps along the curve $K$.
\end{itemize}

\bibliography{LETFnew.bib}
\bibliographystyle{abbrv}

\end{document}